\documentclass[sigconf,authorversion,nonacm]{acmart}

\AtBeginDocument{%
  }

\usepackage{amsmath}
\usepackage{amsthm}
\usepackage{amsfonts}
\usepackage{xspace}
\usepackage{tabularx}
\usepackage[normalem]{ulem} 
\usepackage{bbold}
\usepackage{siunitx}

\usepackage{graphicx}
\usepackage{subcaption}

\usepackage[linesnumbered,ruled]{algorithm2e}
\SetAlFnt{\scriptsize} 
\SetKwBlock{KwOn}{on}{}
\SetKwBlock{KwPeriodically}{periodically}{}
\SetKwBlock{KwState}{state}{}
\SetKwBlock{KwAtomic}{atomically}{}
\SetKwBlock{KwConcurrently}{concurrently}{}
\SetCommentSty{textup} 
\SetKwComment{KwIttayComment}{(}{)} 
\SetKw{KwAnd}{and}
\SetKw{KwOr}{or}
\SetKw{KwSetTimer}{setTimer}
\SetKwBlock{KwExpTimer}{on timer expiration}{}
\SetKwFor{ForAll}{forall}{}

\SetKwBlock{BeginBlock}{}{}{}

\SetKwBlock{GammaOrdRand}{$\ord^\rand$}

\newtheorem{theorem}{Theorem}
\newtheorem{observation}{Observation}
\newtheorem{claim}{Claim}
\newtheorem{lemma}{Lemma}

\newtheorem{definition}{Definition}
\newtheorem{example}{Example}
\newtheorem*{example*}{Example}
\newtheorem{proposition}{Proposition}

\newtheorem{note}{Note}

\newtheorem{construction}{Construction}

\RequirePackage{color}\definecolor{RED}{rgb}{1,0,0}\definecolor{BLUE}{rgb}{0,0,1} 

\newenvironment{restate}[1]{\begin{trivlist} \item {\bf #1 (restated)}. 
  \em} {\end{trivlist}}

\newcommand{\lemmaDeterToBool}{
  Let~$M_\deter$ be a deterministic credential-based mechanism and let the function~$f$ be as constructed above.
  Then~$f$ is monotonic.
  }

\newcommand{\partialBooleanfunction}{
  Let~$M_1$ and~$M_2$ be two mechanisms such that~${\emptyset \neq \Prof{M_1} \subseteq \Prof{M_2}}$.
  Let~$T_1$ and~$T_2$ be the sets of all user availability vectors and~$F_1$ and~$F_2$ be the sets of all attacker availability vectors in~$\Prof{M_1}$ and~$\Prof{M_2}$ respectively.
  Then,~$T_1 \subseteq T_2$ and~$F_1 \subseteq F_2$.
}

\newcommand{\profilePartialFnc}{
  For all mechanisms, there exists an equivalent monotonic partial Boolean mechanism.
}

\newcommand{\numViableScenarios}{    
  The number of viable scenarios for~$n$ credentials is~$4^n-3^n$.
}

\newcommand{\profileSizeBound}{
  Let~$(T,F)$ be a partial Boolean function of~$n$ credentials. 
  Let~$s = max(|T|,|F|)$.
  The~maximum number of scenarios that can be added to the profile of the mechanism~$M_{(T,F)}$ without contradicting~it is
  ${s \cdot (2^n-s)-|\Prof{M_{(T,F)}}|}$ if~$s \geq 2^{n-1}$ and~${4^n-3^n - |\Prof{M_{(T,F)}}|}$ otherwise.
}

\newcommand{\negspace}{\vspace{-0.50\baselineskip}}
\newcommand{\snegspace}{\vspace{-0.25\baselineskip}}

\newcommand{\Msgs}{\ensuremath{\textit{Pending}}\xspace} 
\newcommand{\msg}{\ensuremath{\textit{msg}}\xspace} 
\newcommand{\ord}{\ensuremath{\gamma_\textit{ord}}\xspace} 
\newcommand{\id}{\ensuremath{\gamma_\textit{ID}}\xspace} 
\newcommand{\prof}{\ensuremath{\pi}\xspace} 
\newcommand{\Prof}[1]{\ensuremath{\Pi({#1})}\xspace}
 
\newcommand{\os}{\ensuremath{\textit{OS}\xspace}} 
\newcommand{\cre}{\ensuremath{\textit{cred}\xspace}} 
\newcommand{\deter}{\ensuremath{\textit{det}\xspace}} 
\newcommand{\rand}{\ensuremath{\textit{v}\xspace}} 
\newcommand{\gen}[2]{\ensuremath{\textit{gen}_{#1}({#2})\xspace}} 
\newcommand{\genl}[3]{\ensuremath{\textit{gen}_{#1}^{#2}({#3})\xspace}} 
\newcommand{\step}[2]{\ensuremath{\textit{step}_{#1}({#2})\xspace}} 
 
\newcommand{\processing}{\ensuremath{\textit{processing}\xspace}}

\newcommand{\exec}{\ensuremath{E}\xspace} 
\newcommand{\mloop}{\ensuremath{\textit{Execute}}\xspace} 
\newcommand{\setup}{\ensuremath{\textit{Setup}}\xspace}

\newcommand{\mech}{\text{mechanism}}
\newcommand{\cred}{\text{credential}}

\newcommand{\omga}{\ensuremath{M}\xspace}

\newcommand{\T}{\ensuremath{\textit{T}}\xspace} 
\newcommand{\F}{\ensuremath{\textit{F}}\xspace} 

\newcommand{\suc}{\ensuremath{\textit{suc}}\xspace}

\newcommand{\safe}{\ensuremath{\textit{safe}}\xspace} 
\newcommand{\loss}{\ensuremath{\textit{loss}}\xspace} 
\newcommand{\leak}{\ensuremath{\textit{leak}}\xspace} 
\newcommand{\theft}{\ensuremath{\textit{theft}}\xspace}

\SetCommentSty{mycommfont}

\newcommand{\viableScenarios}{\ensuremath{\textit{viableScenarios}}\xspace} 
\newcommand{\indexx}{\ensuremath{\textit{idx}}\xspace}
\newcommand{\maxSuccessProb}{\ensuremath{\textit{maxSuccessProb}}\xspace} 
\newcommand{\successProb}{\ensuremath{\textit{successProb}}\xspace} 
\newcommand{\maxTable}{\ensuremath{\textit{maxTable}}\xspace} 
\newcommand{\potentialSuccessProb}{\ensuremath{\textit{potentialSuccessProb}}\xspace} 
\newcommand{\truthTable}{\ensuremath{\textit{truthTable}}\xspace} 
\newcommand{\maxPossibleAdditions}{\ensuremath{\textit{maxAdditionalProb}}\xspace} 
\newcommand{\currProfile}{\ensuremath{\textit{currProfile}}\xspace} 
\newcommand{\possibleScenariosProbabilitiesSorted}{\ensuremath{\textit{possibleScenariosProbabilitiesSorted}}\xspace} 
\newcommand{\numPossibleAdditions}{\ensuremath{\textit{numPossibleAdditions}}\xspace} 
\newcommand{\dontCare}{\ensuremath{\bot}\xspace} 
\newcommand{\prevUserVal}{\ensuremath{\textit{prevUserVal}}\xspace} 
\newcommand{\prevAttVal}{\ensuremath{\textit{prevAttVal}}\xspace} 
 
\newcommand{\tableWithScenario}{\ensuremath{\textit{newTable}}\xspace} 

\newcommand{\scenarioBasedSearch}{\ensuremath{\textit{scenarioBasedSearch}\xspace}}

\setcopyright{acmlicensed}
\copyrightyear{2018}
\acmYear{2018}
\acmDOI{XXXXXXX.XXXXXXX}

\acmConference[Conference acronym 'XX]{Make sure to enter the correct
  conference title from your rights confirmation emai}{June 03--05,
  2018}{Woodstock, NY}
\acmISBN{978-1-4503-XXXX-X/18/06}

\begin{document}

\title{Asynchronous Authentication}

\author{Marwa Mouallem}
\email{marwamouallem@cs.technion.ac.il}
\orcid{0009-0008-6334-1183}
\affiliation{%
  \institution{Technion}
  \city{Haifa}
  \country{Israel}
}

\author{Ittay Eyal}
\email{ittay@technion.ac.il}
\orcid{0000-0001-7595-2258}
\affiliation{%
  \institution{Technion}
  \city{Haifa}
  \country{Israel}
}

\begin{abstract}
A myriad of authentication mechanisms embody a continuous evolution from verbal passwords in ancient~times to~contemporary multi-factor authentication: 
Cryptocurrency wallets~advanced from a single signing key to using a handful of well-kept credentials, and for online services, the infamous~``security questions'' were all but abandoned. 
Nevertheless, digital asset heists and numerous identity theft cases illustrate the urgent need to revisit the fundamentals of user authentication. 

We abstract away credential details and formalize the general, common case of \emph{asynchronous authentication}, with unbounded message propagation time. 
Given credentials' fault probabilities (e.g., loss or leak), we seek mechanisms with maximal success probability. 
Such analysis was not possible before due to the large number of possible mechanisms.
We show that every mechanism is dominated by some \emph{Boolean mechanism}---defined by a monotonic Boolean function on presented credentials. 
We present an algorithm for finding approximately optimal mechanisms by leveraging the problem structure to reduce complexity by orders of magnitude.

The algorithm immediately revealed two surprising results: 
Accurately incorporating easily-lost credentials improves cryptocurrency wallet security by orders of magnitude. 
And novel usage of (easily-leaked) security questions improves authentication security for online services. 
\end{abstract}

\maketitle

\section{Introduction}


Authentication is a cornerstone of access control and security~\cite{Auth_from_passwords_to_public_keys, computer_security_principles_and_practice}. 
From ancient times~\cite{eaton_2011} to the early days of the Internet, passwords were often sufficient~\cite{rejecting_the_death_of_passwords}. 
As stakes grow higher with the adoption of online services, companies
deploy advanced \emph{authentication mechanisms}~\cite{MFA,SmartCustody}, and abandon tools like the easily guessable ``security questions'' like ``name of first pet.'' 
%
%
%
Nascent technologies raise the bar further---in cryptocurrencies, in the event of credential loss or leak, there is no central authority that can assist in asset recovery. 
Users thus choose advanced solutions (e.g.,~\cite{Argent}), utilizing multiple well-kept cryptographic keys. 
And even some web services (e.g., the U.S. government website~\cite{us_gov}) require multi-factor authentication, while allowing users to combine various credentials like passwords, fingerprints, face recognition, text messages, security keys and more;
thus, users choose their own authentication mechanism.
%
%
%
%
%
%
%
%
%
%
At times, the number of credentials is larger than what the user is aware of. 
For example, 
online services (e.g., Google~\cite{google_location,google_cookies}) use a device's location, IP address or cookies as additional credentials.
Nevertheless, identity theft incidents are frequent~\cite{ftc2022consumer}, 
so is loss of Bitcoin~\cite{chainalysis2020btcUse, btcloss}, as well as digital asset theft, including a recent heist valued at~\$600M due to a single company's key mismanagement~\cite{cointelegraph2022aftermath}. 
Those illustrate the urgent need to revisit the fundamentals of user authentication. 

Previous work (\S\ref{sec:prev_work}) proposed and deployed various authentication solutions.
These encompass numerous \emph{credential} types, like passwords, biometrics, and smart cards~\cite{SoK_The_Quest_to_Replace_Passwords}. 
But no single credential can be completely relied on~\cite{keep_on_rating,Review_Reports_on_User_Auth_in_Cyber_Security}, leading to the development of advanced mechanisms such as multi-factor authentication~(MFA) and interactive protocols~\cite{MFA,SmartCustody,Argent}. 
However, these works leave open the question of how to model and compare mechanisms' security. 
Eyal~\cite{walleddesign} analyzes a specific type of mechanism, defined by monotonic Boolean functions on presented \emph{credentials} for a small number of credentials.
Maram et al.~\cite{cryptoeprint:2022/1682} study the authentication problem, but with a bounded-delay communication channel between the mechanism and the user---which is often not available.
E.g., when logging in to a web service or issuing small cryptocurrency transactions.

In this work, we study general authentication mechanisms in the common asynchronous setting, with the goal of finding optimally secure mechanisms.
First, we define the \emph{asynchronous authentication problem}~(\S\ref{sec:model}), 
where an authentication mechanism interacts with a \emph{user} and an \emph{attacker} (representing all adversarial entities~\cite{walleddesign, cryptoeprint:2022/1682}) and tries to identify the user.
The mechanism utilizes credentials.
Each credential can be known to the user, the attacker, both, or neither.
The states of all credentials define a \emph{scenario}.
A mechanism is \emph{successful} in a scenario if it correctly identifies the user, for all attacker behaviors and message delivery times. 
The \emph{profile} of a mechanism is the set of all scenarios in which it is successful.


We seek to reason about a general mechanism's success when~message delivery time is unbounded, yet cryptographic assumptions bound running time according to a security parameter.
For simplicity, we first consider ideal credentials that cannot~be forged, following similar abstractions in distributed-systems literature~\cite{van2017distributed,tanenbaum2007distributed, gorilla, safe_permissionless_consensus, Monoxide}. 
We defer a model that considers cryptographic assumptions to the appendix. 
Roughly, we require success when the security parameter~is~large enough, otherwise non-decision is acceptable.
This~approach may be of independent interest for analyzing distributed protocols that rely on cryptography in asynchronous~networks.

To design a mechanism, one must estimate the probabilities of credential faults, e.g., loss or leak~\cite{walleddesign}. 
An optimal mechanism has the maximum probability of success given those estimates. 
To find optimal mechanisms we turn our attention to \emph{Boolean mechanisms}: mechanisms that decide according to a Boolean function of the credentials presented to them. 
For any authentication mechanism, there exists a Boolean mechanism that dominates it---one that succeeds in the same scenarios or more~(\S\ref{sec:mech_to_bool:bool_are_better}).
The proof uses a series of reductions. 
First, interactive message exchanges do not enhance mechanism security---every interactive mechanism is dominated by a non-interactive mechanism 
(assuming as a theoretical construct\footnote{In practice interactivity and randomness are useful, e.g., if multiple steps are required to prove credential availability 
(e.g., challenge-response~\cite{rhee2005challenge} and interactive proofs of knowledge~\cite{feigenbaum1992overview}), 
and for mechanisms usability and deployability~\cite{SoK_The_Quest_to_Replace_Passwords} 
(e.g., a mechanism sending an SMS with a verification code only after verifying a password).} that credential availability can be proven with a single message). 
Second, decisions can be made based solely on credential availability proofs.
Third, even if randomness were utilized, there exists a dominating deterministic one-shot mechanism.
The proofs are constructive and rely on a mechanism's ability to simulate an execution of another. 
Finally, for any deterministic mechanism, there exists a dominating mechanism defined by a monotonic Boolean function of the credentials' availability---a \emph{monotonic Boolean mechanism}. 

Let us thus focus on such mechanisms. 
We show that any monotonic Boolean mechanism defined by a non-constant function is not dominated (maximal), and any maximal mechanism is equivalent (same profile) to a monotonic Boolean mechanism~(\S\ref{sec:mech_to_bool:maximal_mechs}). 
This result shows that studying such mechanisms suffices in the asynchronous setting and reveals a somewhat surprising fact: 
There is no probability-agnostic hierarchy of Boolean mechanisms in the asynchronous setting unlike under synchrony~\cite{cryptoeprint:2022/1682}.
For example, with two credentials~$c_1$ and~$c_2$, consider the mechanism where both~$c_1$ and~$c_2$ are required to authenticate and the mechanism where only~$c_1$ is required.
One can easily verify that neither mechanism dominates the other. 
In fact, this is true in general, for any two non-trivial~$n$-credential Boolean mechanisms.

Since there is no hierarchy of Boolean mechanisms, in general, any such mechanism could be optimal. 
But evaluating all possible mechanisms, as done by Eyal~\cite{walleddesign}, quickly becomes prohibitively complex, as their number grows super-exponentially with the number of credentials~\cite{Dedekind1897}.

Instead, we present a scenario-based algorithm to find approximately optimal mechanisms~(\S\ref{sec:prob}). 
Given~$\delta > 0$, the algorithm finds a \emph{$\delta$\nobreakdash-optimal mechanism}, i.e., a mechanism whose success probability is at most~$\delta$ lower than that of an optimal mechanism. 
The algorithm works greedily by building a mechanism that succeeds in scenarios with the highest probability. 
Once no more scenarios can be added, it removes existing scenarios to add new ones. 
An analysis of different probability distributions shows that it is sufficient to focus on a small number of high-probability scenarios. 
Rather than searching all mechanisms, the algorithm stops when the scenarios that might be added have additional probability lower than~$\delta$. 
This is calculated by bounding the number of scenarios in which a mechanism can succeed. 

Our algorithm allows for finding approximately-optimal mechanisms with many credentials.
This is directly useful for mechanism design. 
For example, when designing a cryptocurrency wallet, users tend to use a handful of high-quality credentials~\cite{Multisig,cold_wallet}, like hardware keys and long memorized mnemonics.
However, we show that adding a few easy-to-lose credentials (e.g., passwords that the user remembers but does not back up) reduces wallet failure probability by orders of magnitude.
Similarly, incorporating easy-to-leak credentials (e.g., the name of the user's pet) with high-quality credentials achieves a similar effect. (Though not in the way such credentials were typically used.) 

In summary~(\S\ref{sec:conclusion}), our contributions are: 
(1) formalization of the asynchronous authentication problem;
(2) proof that every maximal mechanism is equivalent to a monotonic Boolean function;
(3) proof that any two non-trivial monotonic Boolean mechanisms are either equivalent or neither is better, 
(4) an approximation algorithm of optimal mechanisms given the credentials fault probabilities, and 
(5) evaluation of realistic mechanisms and concrete lessons for their improvement. 

\snegspace
\section{Related Work}
\label{sec:prev_work}

Authentication is a fundamental aspect of security~\cite{Auth_from_passwords_to_public_keys, Applied_crypto, security_in_computing, cryptography_and_network_security, computer_security_principles_and_practice} 
that continues to be an active area of research~\cite{fi15040146, sym14040821}. 
Whether explicitly or implicitly, previous work mostly addressed two perspectives: credentials and protocols.
Credentials are the information used to authenticate a user such as passwords~\cite{Password_Security_A_Case_History}, one-time passwords (OTPs)~\cite{OTP}, biometrics~\cite{Biometric}, and physical devices~\cite{FIDO2}.
Protocols define the procedures for authenticating a user based on her credentials. 
Examples include Kerberos~\cite{Kerberos} for devices across a network and cryptocurrency wallets~\cite{coinbase,Ledger,Trezor}.
Our focus is on protocols, particularly the abstract mechanisms behind the protocols, which have received only scant attention. 

Vashi et al.~\cite{Review_Reports_on_User_Auth_in_Cyber_Security} and Bonneau et al.~\cite{SoK_The_Quest_to_Replace_Passwords} analyzed various credential types. 
However, both surveys ultimately concluded that no single credential type offers perfect security, underscoring the importance of protocols that incorporate multiple credentials.
Velásquez et al.~\cite{auth_schemes_and_methods_lit_review_1-2FA} survey single and multi-factor authentication schemes, examining different credentials combined into multi-factor methods.
Zimmermann et al.~\cite{keep_on_rating} rate dozens of deployed authentication methods, focusing on user challenges in credential protection rather than the security of the authentication mechanism. 

Burrows et al.~\cite{a_logic_of_authentication} describe the beliefs of parties involved in authentication protocols as a consequence of communication to realize who has which credentials.
Delegation Logic~\cite{deligation_logic} is a logic-based approach to distributed authorization, 
offering a formal framework for reasoning about delegation relationships.
Neither work compares mechanism security, which is the goal of this work.

The importance of authentication mechanisms has grown significantly due to the increasing value of digital assets, such as cryptocurrencies~\cite{crypto_market} and NFTs~\cite{NFT_market}. 
This is especially true~in decentralized systems, where users are responsible for securing their assets and, unlike centralized systems, no fallback authority exists.
Consequently, numerous cryptocurrency wallets have been developed.
The wide variety of authentication mechanisms used in cryptocurrency wallets highlights the need for a formal way to compare them.
Nevertheless, previous work focused on credential management and {implementation~\cite{SoK_Research_Perspectives_and_Challenges_for_Bitcoin_and_Cryptocurrencies, Bitter_to_Better_How_to_Make_Bitcoin_a_Better_Currency}}. 

Hammann et al.~\cite{access_graphs} uncover vulnerabilities arising from the links between different user credentials and accounts, where one can be used to log into the other.  
The paper focuses on the connection between different credentials and accounts.
However, it does not address the security of different mechanisms, which we do in this work.

Eyal~\cite{walleddesign} analyzes mechanisms that are Boolean functions of credential availability. 
We analyze general mechanisms in an asynchronous network and show that all maximally secure mechanisms in this setting can be reduced to mechanisms that are Boolean functions of credentials availability. 
Eyal finds optimal mechanisms given the probabilities of credential faults for up to~$5$ credentials using a brute force search~and 
bounds the failure probabilities of optimal mechanisms using a heuristic approach.
Although the brute-force search ensures optionality, it is not computationally feasible for more credentials.
We provide an approximation method for near-optimal mechanisms with a large number of credentials~based on our observations on the structure of maximally secure mechanisms. 

Maram et al.~\cite{cryptoeprint:2022/1682} study interactive authentication, where the user and the attacker interact with the mechanism over a synchronous network.
They also introduce security profiles of mechanisms, which we use in our analysis.
However, unlike our work, the paper assumes a synchronous communication model, which is not practical in many cases,
e.g., when processing small payments or logging in to a web service.
In such cases, waiting for the true user to receive a message (e.g., sent by email) is unacceptable.

In contrast to the synchronous case, we show that in the asynchronous case, mechanism security does not benefit from interactivity.
Our results do not imply that interactive mechanisms are not useful. 
On the contrary, are widely used in practice, e.g., in multi-factor authentication~\cite{MFA} and challenge-response protocols~\cite{rhee2005challenge}.
Consider for example Google's security alerts~\cite{google_security} that notify the user of suspicious activity in her account, 
allowing her to confirm or deny the activity. 
This is an interactive asynchronous two-factor authentication mechanism, where the user can respond to the mechanism's request at any time.
An equivalent non-interactive mechanism can be one that requires both the user's password and her fingerprint to authenticate. 
However, in this case, the non-interactive mechanism is less convenient for the user.
For every interactive mechanism, we show a theoretical reduction to a non-interactive mechanism that succeeds in the same scenarios or more.
All maximally secure mechanisms~have~equivalent mechanisms that are Boolean functions of credentials availability, which is not the case in the synchronous model.

Although our analysis uses methods similar to those used in distributed systems' theory like asynchronous communication and simulations, the authentication problem is distinct from classical problems like consensus or broadcast~\cite{algorithmic_foundations_of_differential_privacy}:
Success is defined by the decision of a single party and not multiple ones, and credentials take a central role.

\negspace
\section{Model}
\label{sec:model}
We formalize the asynchronous authentication problem~(\S\ref{sec:model:asynchronous_authentication}), specifying an execution, its participants, and their behavior.
Then we define mechanism success and the relation between different mechanisms~(\S\ref{sec:model:mechanism_success}). 

\negspace
\subsection{Asynchronous Authentication}
\label{sec:model:asynchronous_authentication}
The system comprises three entities, an \emph{authentication mechanism}~$M$ and two \emph{players}: a \emph{user}~$U$ and an \emph{attacker}~$A$. 
All three are finite deterministic automata that can draw randomness from a \emph{random tape}~$v$, an infinite stream of random bits. 
We specify the automata as computer programs~\cite{state_machine_rep}.

An execution is orchestrated by a \emph{scheduler},
whose pseudocode is given in Appendix~\ref{app:execution}.
Each of the players interacts with the mechanism by sending and receiving messages.
The scheduler, parametrized by~$\id\in\{0,1\}$, assigns the user identifier~$\id$ and the attacker identifier~${(1-\id)}$.
E.g., if~$\id = 0$, then the user is player~$0$ and the attacker is player~$1$.
The identifiers serve a similar purpose to that of cookies that identify a website visitor during a single session. 
They allow the mechanism to identify and distinguish between the players during an execution without revealing which is the user. 

To facilitate authentication, the players take advantage of private information to convince the mechanism that they are the user.
This information is a set of credentials that might be data (e.g., passwords), biometric properties (e.g., fingerprints), and physical objects (e.g., phones or smart cards).
Each credential has two parts: a public part, known to the mechanism, and a secret part. 
For example, when using a password to authenticate to a website, the user remembers the password (the secret part) and the website stores its hash (the public part);
or when using public key authentication, a blockchain has the public key (the public part) and the user has the private key (the secret part).

\snegspace
\begin{definition}[Credential]
    \label{credential}
    A \emph{credential}~$c$ is a tuple ${c=(c^P,c^S)}$ where~$c^P$ is the \emph{public part} of the credential and~$c^S$ is the \emph{secret part} of the credential. 
\end{definition}

\snegspace  
The system uses a set of~$n$ \cred s,~$\{c_1,...,c_n\}$.
A \cred~can be \emph{available} to a player. 
An \emph{availability vector} represents the availability of all credentials to a player. 
Each entry~$i$ in the vector indicates the availability of credential~$i$.
Denote by~$\sigma_U$, ~$\sigma_A$ the availability vector of the user and the attacker respectively.
A scenario~$\sigma$ is a tuple of availability vectors~$(\sigma_U,\sigma_A)$.
Denote by~$C_U^{\sigma}$, ~$C_A^\sigma$ the set of \cred s available to the user and attacker in scenario~$\sigma$ respectively.
The attacker is polynomially bounded~\cite{poly_adv1,poly_adv2} and cannot guess the secret parts of the credentials better than a random guess in a polynomial number of steps, except with negligible probability.

The system uses a credential generation function~$\genl{}{}{n}$ that generates a set of~$n$ credentials~$\{(c_i^P,c_i^S)\}_{i=1}^n$. 
Anyone with the secret part of a credential can prove its availability to the mechanism, and anyone without the secret part cannot forge a proof of its availability. 
In practice, cryptographic assumptions require bounding the execution time based on a security parameter~\cite{RSA}. 
However, in the asynchronous system, message delivery time is unbounded, and so is execution time---the mechanism is only required to eventually decide. 

For simplicity, we first assume that credentials are \emph{ideal}, that is, that their cryptographic guarantees are maintained indefinitely. 
This is a common assumption in distributed-systems literature, e.g., \cite{van2017distributed, tanenbaum2007distributed, gorilla, safe_permissionless_consensus, Monoxide}. 
We properly reconcile the tension between the models in Appendix~\ref{app:full_model}; intuitively, the mechanism is only required to succeed when the security parameter is large enough.
The results hold for this model as well, as shown in Appendix~\ref{app:one_shot_dominance_full_proof}.

An \emph{authentication mechanism} exchanges messages with the players to decide which is the user.
\begin{definition}[Authentication mechanism]
    An \emph{authentication mechanism}~$M$ is a finite deterministic automaton specified by two functions,~$\gen{M}{\cdot}$ and~$\step{M}{\cdot}$: 
    \label{Authentication mechanism}
    \begin{itemize}
        \item $\gen{M}{n}$, where~$n$ is the number of credentials, is an ideal credential generator.  And
        \item $\step{M}{\msg, i}$, where~$\msg$ is the message the mechanism received (perhaps no message, represented by~${\msg=\bot}$), sent from player~$i$ (if no message was received,~$i=\bot$ ), 
        is a function that updates the state of the mechanism and returns a pair~$(\msg s_0,\msg s_1)$ of sets of messages (maybe empty) to send to the players by their identifiers.
        It can access and update the mechanism's state and the set of credentials public part~$\{c_i^P\}_{i=1}^n$.
    \end{itemize}
    The mechanism uses a variable~$\textit{decide}_{M}$ that marks which identifier it decides belongs to the user.
    The variable's initial value is~$\bot$.
    When~$M$ reaches a decision, it updates the variable to an identifier value. 
\end{definition}

\snegspace
A \emph{player} is defined by its \emph{strategy}, a function that defines its behavior. Formally,
\begin{definition}[Strategy]
    \label{strategy}
     A \emph{strategy}~$S$ of a player~$p\in\{U,A\}$ is a function that takes a message from the mechanism as input
    (maybe empty, denoted by~$\bot$), and has access to the player's state and credentials~$C_p^\sigma$. 
    It updates the player's state and returns a set of messages to be sent back to the mechanism.
\end{definition}

We now describe the execution of the system. 
Full details including pseudocode are in Appendix~\ref{app:execution}.
The execution consists of two parts: \emph{setup} (\setup function) and a \emph{main loop} (\mloop function).
The setup generates a set of~$n$ credentials using the function{~$\gen{M}{n}$}, and assigns each player credentials according to the scenario~$\sigma$.
Then it assigns an identifier to each player based on~$\gamma_\textit{ID} \in \{0,1\}$.
Whenever~a~random coin is used by~$M$ or the players, the result is the next bit~of~$v$. 

Once the setup is complete, the main loop begins and both players~$p\in\{U,A\}$ can send messages to the mechanism, each based on her strategy~$S_p$ and set of credentials~$C_p^\sigma$. 
Time progresses in discrete steps.
The communication network is reliable but asynchronous.
Messages arrive eventually, but there is no bound on the time after they were sent and no constraints on the arrival order. 

This is implemented as follows.
The scheduler maintains three sets of pending messages, one for each entity.
The scheduler is parametrized by an ordering function~$\ord$ and a random tape~$v_\gamma$.
In each step, it uses~$\ord$ to choose a subset of messages (maybe empty) from each of the pending message sets, and returns them as an ordered list.
The scheduler's random coins used in \ord are taken from~$v_\gamma$.
The scheduler removes the chosen messages from the pending message sets.

Each player receives the messages chosen by the ordering function~$\ord$ and sends messages to the mechanism according to her strategy. 
These messages are added to the mechanism's pending messages set.
Similarly, the mechanism receives the messages chosen by~$\ord$. 
After every message~$\msg$ it receives from the player with identifier~$i\in\{0,1\}$, the mechanism runs its function~$\step{M}{\msg, i}$ and sends messages to the players.

In each step, the scheduler checks if~$\textit{decide}_{M} \neq \bot$, meaning that the mechanism has reached a decision. 
Once it does, the execution ends and the player with the matching index (either~$0$ or~$1$) wins.
The tuple~$(\id,\ord, v_\gamma)$ thus defines the \emph{scheduler}'s behavior. 
And an execution~$\exec$ is thus defined by its parameter tuple~$(M, \sigma, S_U, S_A, \gamma, v)$; by slight abuse of notation we write~${\exec = (M, \sigma, S_U, S_A, \gamma, v)}$.



We define the \emph{winner} of an execution as the player with the identifier that the mechanism decides in that execution.

\let\oldnl\nl
\newcommand{\nonl}{\renewcommand{\nl}{\let\nl\oldnl}}


\negspace
\subsection{Mechanism Success}
\label{sec:model:mechanism_success}

A mechanism is \emph{successful} in a scenario~$\sigma$ if the user wins against all attacker strategies and schedulers. 
Formally, 
\begin{definition}[Mechanism success]
    \label{success}
    A mechanism~${M}$ is \emph{successful} in a scenario~$\sigma$ if there exists a user strategy~$S_U$ such that for all attacker strategies~$S_A$, 
    schedulers~$\gamma$, and random tapes~$v$, the user wins the execution~$\exec = (M, \sigma, S_U,S_A,\gamma, v)$. 
    Such a user strategy~$S_U$ is a \emph{winning user strategy} in~$\sigma$ with~$M$.
    Otherwise, the mechanism \emph{fails}. 
\end{definition}

The set of scenarios in which the mechanism succeeds~is~its \emph{profile} (as Maram et al.~\cite{cryptoeprint:2022/1682} defined for a synchronous network).

\begin{definition}[Profile~\cite{cryptoeprint:2022/1682}]
    \label{profile} 
    A \emph{profile} is a set of scenarios. 
    The \emph{profile of a mechanism}~$M$, denoted~$\Prof{M}$, is the set of all scenarios in which~$M$ succeeds. 
\end{definition}

We are now ready to define the \emph{asynchronous authentication problem}.

\snegspace
\begin{definition}[Asynchronous authentication]
    \label{async_auth_problem}
    Given a profile~$\prof$, a mechanism~$M$ solves the \emph{$\prof$ asynchronous authentication problem} if~$M$ is successful in all scenarios in~$\prof$ (maybe more), i.e.,~${\prof \subseteq \Prof{M}}$.
\end{definition}

\snegspace
In the asynchronous setting, the profile defines a relation between any two mechanisms~$M_1$ and~$M_2$ with the same number of credentials.
The following definition is similar to that given by Maram et al.~\cite{cryptoeprint:2022/1682} for a synchronous network, but as we will consider individual credential probabilities, our definition is permutation sensitive. 

\snegspace
\begin{definition}[Mechanisms order]
    \label{partial ordering of mechanisms}
    A mechanism~$M_1$ \emph{dominates} (resp., \emph{strictly dominates}) a mechanism~$M_2$ with the same number of credentials
    if the profile of~$M_1$ is a superset (resp., strict superset) of the profile of~$M_2$: $\Prof{M_2}\subseteq \Prof{M_1}$ (resp.,~$\Prof{M_2}\subset \Prof{M_1}$). 
    
    Two mechanisms~$M_1$ and~$M_2$ are \emph{equivalent} if they have the same profile~$\Prof{M_1}=\Prof{M_2}$.
    And \emph{incomparable} if neither dominates the other~$\Prof{M_1}\nsubseteq \Prof{M_2}$ and~${\Prof{M_2}\nsubseteq \Prof{M_1}}$.
\end{definition}

\negspace
\section{Maximal Mechanisms}
\label{sec:mech_to_bool}

Having defined partial ordering on mechanisms, we proceed to identify \emph{maximal mechanisms}, i.e., mechanisms that are not strictly dominated. 
We show that for any mechanism there exists a dominating mechanism that is a Boolean function of the credentials' availability~(\S\ref{sec:mech_to_bool:bool_are_better}).
Then we show that mechanisms defined by monotonic Boolean functions are maximal, and all maximal mechanisms are equivalent to a Boolean mechanism~(\S\ref{sec:mech_to_bool:maximal_mechs}).  

\snegspace
\subsection{Domination by Boolean Mechanisms}
\label{sec:mech_to_bool:bool_are_better}

We show that any mechanism is dominated by a mechanism that is a Boolean function of the credentials' availability.
We take 4 similar steps.
First, we show that for all mechanisms there exists a dominating \emph{one-shot} mechanism that decides based on up to a single message from each player~(\S\ref{sec:mech_to_bool:bool_are_better:one_shot}). 
Second, we show that decisions can be made solely on credential availability proofs~(\S\ref{sec:mech_to_bool:bool_are_better:cred}). 
Then we show that randomness does not improve the security of a one-shot mechanism~(\S\ref{sec:mech_to_bool:bool_are_better:deter}). 
Finally, we show that for any one-shot, deterministic mechanism, there exists a dominating mechanism defined by a monotonic Boolean function of the credentials' availability~(\S\ref{sec:mech_to_bool:bool_are_better:bool}). 

\snegspace
\begin{note}[Practicality of one-shot and deterministic mechanisms]
    \label{note:practicality_of_one_shot}
    Numerous interactive and randomized mechanisms are practical and widely used, 
    e.g., challenge-response protocols~\cite{rhee2005challenge} and interactive proofs of knowledge~\cite{feigenbaum1992overview}. 
    While we prove a reduction from any mechanism to a one-shot and deterministic mechanism,
    the propositions below are for theoretical purposes and 
    do not necessarily imply that such mechanisms are practical nor suggest that they should replace interactive, randomized mechanisms.
    
    We rather show that, security-wise, for every interactive or random mechanism, 
    there exists a non-interactive deterministic mechanism that succeeds in the same scenarios or more.
    This is only used as a proof step of the reduction to monotonic Boolean functions.
    However, the non-interactive mechanism might be less practical.
    E.g., a 2FA mechanism that first requires a password and only when the password is correct, 
    it asks for a one-time password sent to the user's phone, 
    is more practical than a mechanism that requires the user to send both the password and the one-time password in a single message, 
    as it saves the company the cost of sending the one-time password in case the password is incorrect.
\end{note}

\negspace
\subsubsection{Step template}
Each step in our proof has the following parts:
Given a mechanism~$M_1$ we show it is dominated by a mechanism~$M_2$ with a certain property.
We prove constructively by defining a mechanism~$M_2$ and show the required property holds. 
The domination proof takes advantage of a mechanism's ability to simulate the execution of another mechanism.
That is, the mechanism~$M_2$ takes a mechanism~$M_1$, two strategies~$S_0$ and $S_1$, the sets of credentials each strategy uses~$C_0$ and~$C_1$, an ordering function~$\ord$, a scheduler random tape~$v_\gamma$, a random tape~$v$, and a number of steps~$t$,
and runs~$\mloop(M_1, S_0, C_0, S_1, C_1, \ord, v_\gamma, v, t)$ (Appendix~\ref{app:execution}). 
Note that when calling the function~$\mloop$, if~$t$ is not explicitly given, it is assumed to be unbounded.
The simulation result is the identifier that~$\textit{decide}_{M_1}$ gets during the execution of~$\mloop$, if it terminates with a decision, and~$\bot$ otherwise.

\begin{note}[Mechanism's computational resorces] \label{note:mechanism_power}
    In the following proofs, we assume the constructed mechanism has sufficient computational resources to simulate the original mechanism in a single step.
    At the end of this section, we show that there always exists a dominating polynomial mechanism.
\end{note}

\begin{note}[Plaintext credentials] \label{note:plaintext_creds}
    In the following proofs, we assume the user sends the secret part of her credentials (as with passwords). 
    Although this is not always the case (e.g., with cryptographic signatures), 
    this is only a theoretical construction and does not suggest that the user should generally send her secrets in plaintext. 
    In practice, she must only be able to prove she can access them.
\end{note}

\begin{note}[Mechanism as a function] 
    When constructing a mechanism~$M_\textit{new}$ given a mechanism~$M$, 
    we use the notation~$M_\textit{new}(M)$ to denote the mechanism~$M_\textit{new}$ is a function of~$M$.
    When it is clear from the context, we omit the argument~$M$ and write~$M_\textit{new}$.
\end{note}

\negspace
\subsubsection{One-shot Mechanisms}
\label{sec:mech_to_bool:bool_are_better:one_shot}
We define a \emph{one-shot} (OS) mechanism as one that, in every execution, reaches a decision based only on the first message it receives (if any) from each player.
If a player sends multiple messages, the mechanism ignores all but the first it receives.

\begin{definition}[One-shot mechanism]
    \label{def:one_shot}
    A mechanism~$M$ is a \emph{one-shot mechanism} if for all scenarios~$\sigma$, schedulers~$\gamma$, and random tapes~$v$,
    and for all user strategies~$S_U, S_U'$ and attacker strategies~$S_A$,~$S_A'$ 
    such that the first message that~$M$ receives from each player (if any) in the executions ${\exec_1 = (M, \sigma, S_U, S_A, \gamma, v)}$ 
    and~${\exec_2 = (M, \sigma, S_U', S_A', \gamma, v)}$ is the same and in the same order, 
    then~$M$ reaches the same decision or does not decide in both executions.
\end{definition}

\begin{algorithm}[t]
    \SetInd{0.1em}{0.5em}
    \SetAlgoNoLine 
    \SetAlgoNoEnd 
    \DontPrintSemicolon 
        
    \SetKwInOut{Input}{input}\SetKwInOut{Output}{output}
    \nonl \small \emph{$\step{M_\os}{\msg, i}$} \scriptsize \;

    \For {$j \in \{0,1\}$}{
        \If(\tcp*[f]{A message already received from~$j$}){$\processing_j = 1$}{ \label{line:message_already_recieved}                         
            $(S_j, C_j, v_{\gamma,j}, v'_j) \gets \textit{Sim}_j$         \label{line:set_strategy_mem}                         \tcp*{Read simulation params}
            $S_{1-j}, C_{1-j} \gets \bot, \bot$\;               \label{line:set_strategy_end_mem}
            $\textit{decide}_{M_\os} \hspace{-0.25em} \gets \hspace{-0.25em} \mloop(M, S_0, C_0, S_1, C_1, \ord^\rand, v_{\gamma,j}, v'_j,t)$ \label{line:simulate_M_mem}                         \tcp*{Simulate $M$'s execution}
        }                     
        \ElseIf(\tcp*[f]{First message received from~$i$}){$j = i$}{                                            
            $\processing_i \gets 1$ \label{line:message_already_recieved_end}                         \tcp*{Mark receiving a message from~$j$}
            $S_i, C_i\gets \textit{extractStrategy}(\msg)$  \label{line:set_strategy}                         \tcp*{Extract strategy}
            $S_{1-i}, C_{1-i} \gets \bot, \bot$\;               \label{line:set_strategy_end}

            \If(\tcp*[f]{Invalid strategy}){$S_i = \bot$}{
                $\textit{decide}_{M_\os} \gets 1-i$  \label{line:invalid_S}                         \tcp*{The other player wins}
            }
            \Else{
                $\textit{Sim}_i \gets (S_i, C_i, v_\gamma, v')$    \tcp*{Save simulation params} \label{line:save_strategy_mem}
                $\textit{decide}_{M_\os} \gets \mloop(M, S_0, C_0, S_1, C_1, \ord^\rand ,v_\gamma, v',t)$ \label{line:simulate_M}                         \tcp*{Simulate $M$'s execution}
            }   
        }
    }
    $\textit{return } \bot , \bot$\  \tcp*{No messages to send} 

    \caption{$M_\os(M)$'s step function}
\label{alg:M_{OS}}
\end{algorithm}

        





Given any mechanism, we construct a dominating one-shot mechanism by simulating the original mechanism's execution. 

\begin{construction} \label{con:M_os}
We define~$M_\os$ by specifying the functions~$\gen{M_\os}{\cdot}$ and $\step{M_\os}{\cdot}$.
The credentials' generation function~$\gen{M_\os}{\cdot}$ is the same as~$\gen{M}{\cdot}$.

The mechanism's \textit{step} function proceeds as follows (Algorithm~\ref{alg:M_{OS}}.) 
If it does not receive a message during its execution, it does nothing.
If it receives multiple messages from a player, it ignores all but the first one~(lines~\ref{line:message_already_recieved} and~\ref{line:message_already_recieved_end}).
If~$M_\os$ receives a message that is not a valid strategy and credentials pair, then it decides the identifier of the other player~(line~\ref{line:invalid_S}).

Consider the first message it receives, and let~$t_1 \in \mathcal{N}^+$ be the step in which it arrives.
If the message is an encoding of a valid strategy and credentials pair, then~$M_\os$ simulates an execution of~$M$~(line~\ref{line:simulate_M}) 
with the given strategy and credentials while setting both the opponent's strategy and credentials to~$\bot$ each~(lines~\ref{line:set_strategy}-\ref{line:set_strategy_end}).
It uses a scheduler random tape~$v_\gamma$ and an execution random tape~$v'$ drawn from~$v$, and an ordering function~$\ord^{v_\gamma}$ 
that chooses the time and order of message delivery randomly based on~$v_\gamma$.
The simulation runs for~$t_1$ steps.
If~$M$'s simulated execution decides then~$M_\os$ decides the same value.
Otherwise,~$M_\os$ saves the above execution details~(line~\ref{line:save_strategy_mem}) and waits for the next message.

In each time step~$t'>t_1$ between the message arrival from the first player and the message arrival from the second 
player (might be infinite),~$M_\os$ reads the simulation parameters it saved~(line~\ref{line:set_strategy_mem}) and runs the same simulation with the exact same parameters as before but for~$t' \in \{t_1 +1, t_1 +2, ...\}$ steps each time~(line~\ref{line:simulate_M_mem}).
If the simulation reaches a decision, then~$M_\os$ decides the same value.

If a message arrives from the other player at some time~$t_2 > t_1$, then, similar to the previous case,~$M_\os$ simulates an execution of~$M$ 
with the given strategy and credentials while setting the opponent's to~$\bot$ for~$t_2$ steps.
And again, if~$M$'s simulated execution decides, then~$M_\os$ decides the same value.

Otherwise,~$M_\os$ saves the above execution details~(line~\ref{line:save_strategy_mem}) and continues to simulate the execution of~$M$ for~$t''\in \{t_2 +1, t_2 +2, ...\}$ steps.
In each of its steps,~$M_\os$ runs two simulations of~$M$'s execution, 
one with the first player's simulation parameters and the other with the second player's simulation parameters, both for~$t''$ steps~(line~\ref{line:simulate_M_mem}). 
Once a simulation decides, then~$M_\os$ decides the same value.

\end{construction}

The mechanism~$M_\os(M)$ is one-shot as it decides based only on the first message it receives from each player.
To prove it dominates~$M$, we first show that the attacker's message does not lead to her winning.

\begin{lemma}
    \label{lem:M_os_attacker_cant_win}
    Let~$\sigma$ be a scenario and let~$M$ be a mechanism successful in~$\sigma$.
    Then, for all executions of~$M_\os(M)$ in scenario~$\sigma$ in which the function~$\step{M_\os(M)}{\cdot}$ receives a message from the attacker for the first time, 
    either the function sets~$\textit{decide}_{M_\os}$ to the user's identifier or the simulated execution of~$M$ does not decide.
\end{lemma}

\begin{proof}
    Let~$\sigma$ be a scenario and let~$M$ be a mechanism successful in~$\sigma$.
    Consider an execution of~$M_\os$ in~$\sigma$ in which the function~$\step{M_\os}{\cdot}$ receives an attacker's message for~the first time.
    Let~$t$ be the time step in which the message arrives and denote the identifier of the attacker by~$i\in \{0,1\}$.
    If the attacker's message is not a valid encoding of a strategy and credentials set, then~$M_\os$ decides the identifier of the user, and we are done. 
    
    Otherwise, the attacker's message is an encoding of a valid strategy and credentials pair~$(S_A,~C_A)$.
    In this case,~$M_\os$ sets the strategies and credentials to~$S_i=S_A$,~$C_i=C_A$,~$S_{1-i}= \bot$, and~${C_{1-i}=\bot}$, $\ord^{v_\gamma}$,~$v_\gamma$ and~$v'$ as in the definition of~$M_\os$, 
    and simulates~$M$'s execution by running~$\mloop(M, S_0, C_0, S_1, C_1, \ord^{v_\gamma}, v_\gamma, v', t)$. 
    If the simulation returns the user's identifier or does not decide for any~$t$, we are done. 
    The only remaining option is that the simulation returns the identifier of the attacker.
    To show this is impossible, assume by contradiction that there exists a~$t' \geq t$ for which the attacker wins in this simulated execution.
    
    First, we show that there exists an execution of~$M$ that is identical to the simulated one. 
    Let~${\gamma=(\id^\os,\ord^{v_\gamma}, v_\gamma)}$ be~a scheduler such that the user gets the same identifier as in~the execution of~$M_\os$, with the same ordering function and scheduler random tape that~$M_\os$ uses to simulate~$M$'s main loop.
    And let~$v_M$ be a random tape such that when the execution $\exec=(M,\sigma,\bot,S_A,\gamma, v_M)$ reaches the main loop, all the next bits of $v_M$ are equal to~$v'$.
    Note that the main loop of~$\exec$ is the same as the one~$M_\os$ simulates in~$\exec_\os$.
    And the attacker wins in the execution~$\exec$ (by the contradiction assumption).

    We thus established a simulated execution in which the attacker wins, and in this simulation, there exists a time~$\tau \leq t'$ when the simulated~$M$ decides the identifier of the attacker.  
    Since~$M$ is successful in scenario~$\sigma$, there exists a winning user strategy~$S_U$.
    Let~$\gamma'$ be a scheduler such that in the execution~$\exec' = (M,\sigma, S_U, S_A, \gamma',v_M)$ it behaves like~$\gamma$ except~$M$ receives the user's messages not before~$\tau$.
    Such a scheduler exists since the communication is asynchronous and message delivery time is unbounded.
    
    At any time step before~$\tau$, the mechanism~$M$ sees the same execution prefix whether it is in the execution~$E$ with an empty user strategy or in the execution~$E'$ with a winning user strategy.
    Thus, it cannot distinguish between the case where it is in~$\exec$ or~$\exec'$ at~$\tau$. 
    Since in~$\exec$ the mechanism~$M$ decides at~$\tau$,~$M$ must decide the same value at~$\tau$ in~$\exec'$.
    That is, the attacker wins also in the execution~$\exec'$, contradicting the fact that~$S_U$ is a winning user strategy for~$M$ in~$\sigma$.  
    Thus, the attacker cannot win in the simulated execution of~$M$.
\end{proof}

Now we can prove domination.

\begin{proposition}
    \label{thm:oneShot}
    For all profiles~$\prof$, an authentication mechanism that solves the~$\prof$~asynchronous authentication problem is dominated by a one-shot mechanism. 
\end{proposition}

\begin{proof} 
    Assume that~$M$ is successful in a scenario~$\sigma$.
    Then, there exists a user strategy~$S_U$ such that for every attacker strategy~$S_A$, scheduler~$\gamma$, and random tape~$\widetilde{v}$, the user wins the corresponding execution~$(M, \sigma, S_U, S_A, \gamma, \widetilde{v})$.

    We now show that~$M_\os$ is successful in scenario~$\sigma$ as well.
    Denote by~$S_U^\os$ the user strategy that sends an encoding of the winning user strategy~$S_U$ and a set of credentials~$C_U\subseteq C_U^\sigma$ that~$S_U$ uses in a single message on the first step.
    And consider any execution of~$M_\os$ in scenario~$\sigma$ with the user strategy~$S_U^{\os}$.
    Note: In the~$\setup$ function of an execution, the first player in the tuple is always the user and the second is the attacker. 
    However, in the~$\mloop$ function, the first player is the player with identifier~$0$, which can be the user or the attacker, and the second is the player with identifier~$1$.
    
    As~$S_U^{\os}$ sends a message in the first step, the user's message eventually arrives at some time step~$t$. 
    The mechanism receives and processes at least one message during its execution; we consider two cases separately: a user's message arrives first, or an attacker's message arrives first.

    Denote by~$i\in \{0,1\}$ the identifier of the player whose message arrives first and by~$t_1$ its arrival time.
    If the user's message~$\msg$ arrives first with the encoded strategy and credential set~${(S_U,~C_U) = \textit{extractStrategy}(\msg)}$,~then~$M_\os$~sets the strategies to~$S_i=S_U$,~$C_i=C_U$,~${S_{1-i}=\bot}$,~and~$C_{1-i}=\bot$, the scheduler random tape~$\gamma_v$,~the~ordering~function~$\ord^{v_\gamma}$, and the random tape~$v'$ as described in~$M_\os$'s definition,
    and simulates~$M$'s execution by~running $\mloop(M, S_0, C_0, S_1, C_1, \ord^{v_\gamma},v_\gamma, v', t_1)$. 
    Denote by~$\gamma=(\id^\os,\ord^{v_\gamma}, v_\gamma)$ the scheduler such that the user gets the same identifier as in the execution of~$M_\os$, with the same ordering function and scheduler random tape that~$M_\os$ uses to simulate~$M$'s main loop.
    Because~$S_U$ is a winning user strategy in~$M$, the user wins the execution~$\exec = (M, \sigma, S_U, \bot, \gamma, v')$.
    Therefore, there exists a time step~$\tau$ such that~$M$ decides the user at~$\tau$.
    If~$t_1 \geq \tau$, then~$M_\os$ decides the user as well.
    Otherwise,~$t_1 < \tau$, then~$M_\os$ repeatedly simulates~$M$'s execution for~$t' \in \{t_1+1 ,t_1+2, ...\}$ steps.
    If no other message arrives before~$\tau$, then once it reaches~$t'=\tau$,~$M_\os$ decides the user as well.
    Otherwise, if the attacker's message arrives at some~$t_2 < \tau$, 
    and by Lemma~\ref{lem:M_os_attacker_cant_win}, either~$M_\os$ decides the user or the simulated execution of~$M$ with the attacker's strategy and credentials does not decide.
    Again, if~$M_\os$ decides the user, then we are done.
    Otherwise,~$M_\os$ continues to simulate both executions of~$M$ with the user's and attacker's simulations parameters, respectively, for~$t'' \in \{t_2+1, t_2+2, ...\}$ steps.
    Finally, at~$t''=\tau$, the simulated execution of the user's strategy decides the user, so~$M_\os$ decides the user as well.

    Now assume the attacker's message arrives first to~$M_\os$ at time~$t_1$.
    Again by Lemma~\ref{lem:M_os_attacker_cant_win}, either~$M_\os$ decides the user or the simulated execution of~$M$ in~$\sigma$ with the attacker's strategy does not decide.
    If~$M_\os$ decides the user, then we are done.
    Otherwise, the simulation does not decide for all~$t\geq t_1$, and~$M_\os$ waits for the next message from the other player (the user).
    As long as no other message arrives,~$M_\os$ keeps simulating the same execution in each step for a longer time but does not reach a decision. 

    Because the user's message must eventually arrive,~$M_\os$ receives the user's message with her encoded strategy and credentials at time~$t_2 > t_1$. 
    Then similar to the previous case,~$M_\os$ sets the parameters as described in Construction~\ref{con:M_os},
    and simulates~$M$'s execution. 
    By the same argument as before, we get that once~$t_2 \geq \tau$, the simulation terminates and returns the user's identifier, so~$M_\os$ also returns it. 

    Overall we showed that the user wins in all executions of~$M_\os$ in scenario~$\sigma$.
    Thus, it is successful in~$\sigma$.
    And we~conclude that the one-shot mechanism~$M_\os$ dominates~$M$.
\end{proof}

We illustrate Proposition~\ref{thm:oneShot} with an example.
\begin{example}
    \label{ex:oneShot}
    Consider a mechanism~$M^\textit{ex}$ with two credentials: a password~$c_1$ and an OTP~$c_2$.
    Authentication with~$M^\textit{ex}$ has three steps:
    First, the player clicks "start". 
    Then, she enters her username and password. 
    If those are correct, she is asked to enter the OTP.
    If all three steps are done successfully, the player wins and gets authenticated.
    In case no player wins, the mechanism chooses a player at random.

    For any player to authenticate, she must complete all three steps successfully.
    This not only requires knowing the correct credentials but also requires that all messages arrive before the mechanism decides.
    However, even if the user sends the correct credentials, some of the messages might be delayed, 
    and as the mechanism must decide, it might do so before receiving all messages.
    In this case, the mechanism chooses a player at random, and the user might not authenticate.
    Therefore, this mechanism's profile is empty~$\Prof{M^\textit{ex}} = \emptyset$.

    A dominating one-shot mechanism~$M_\os^\textit{ex}$ has the player send the content of all three steps in a single message.
    The rest of the decision process of the mechanism stays the same.
    If the player sends the content of all three steps in a single message, then she wins.
    Therefore, if the user knows both credentials, she can authenticate.
    However, if the user knows only one of the credentials, she cannot always authenticate.
    And if the attacker knows at least one of the credentials, she can authenticate in some cases (depending on the mechanism's random choice).
    Therefore, the mechanism's profile is the scenario in which the user knows both credentials and the attacker knows none~$\Prof{M_\os^\textit{ex}} = \{((1,1),(0,0))\} \supset \Prof{M^\textit{ex}}$.
\end{example}


\negspace

\subsubsection{Credential-based mechanisms}
\label{sec:mech_to_bool:bool_are_better:cred}

A one-shot mechanism is a \emph{credential-based} mechanism if it reaches a decision based only on a function of the credentials it receives from players.

\begin{definition}[Credential-based mechanism]
    A mechanism~$M$ is a \emph{credential-based} mechanism if it is a one-shot mechanism and for all scenarios~$\sigma$, 
    schedulers~$\gamma$ and random tapes~$v$, 
    and for all user and attacker strategies~$S_U, S_U', S_A$ and~$S_A'$ such that~$S_U$ and~$S_A$ send only messages that are subsets of the corresponding player's credentials' vector,
    and~$S_U'$ and~$S_A'$ send messages with the same subsets of credentials but with any additional strings,
    then~$M$ reaches the same decision in both executions~$\exec_1=(M, \sigma, S_U, S_A, \gamma, v)$ and~$\exec_2=(M, \sigma, S_U', S_A', \gamma, v)$.
\end{definition}

\begin{lemma}
    \label{lem:onlyCredsHelp}
    For all one-shot mechanisms~$M_\os$ there exists a credential-based mechanism~$M_\cre$ that dominates~$M_\os$.
\end{lemma}

The proof of Lemma~\ref{lem:onlyCredsHelp} is similar to that of Proposition~\ref{thm:oneShot} and is omitted for brevity.
The key difference is the observation that any message a player can send is a function of (1)~common information \cite{common_knowledge} available to everyone, 
(2)~credentials of the sender, and (3)~random bits.

Denote by~$M_\os$ a one-shot mechanism.
If a player has a winning strategy for~$M_\os$, knowing some common information, her credentials, and the random bits, she can generate the messages she needs to send to win in the credential-base mechanism as well.

Formally, we define a function~$S_\os(\cdot)$ that maps a set of credentials to a strategy.
Let~$c$ be a set of credentials known to player~$i\in \{0,1\}$, the function~$S_\os(c)$ returns a strategy for player~$i$ for~$M_\os$ using~$c$ such that if the first message that~$M_\os$ receives is from player~$i$, it leads to player~$i$ winning the execution. 
If no such a strategy exists, $S_\os(c)=\bot$.
Then the construction of a credential-based mechanism~$M_\cre$ is similar to that of~$M_\os$ 
with the following change:
Instead of extracting the strategy and credentials from the message,~$M_\cre$ uses the function~$S_\os(\cdot)$ to obtain the strategy given the credentials.

The mechanism~$M_\cre$ is credential-based as is one-shot and decides based only on the credentials it receives.

\begin{example}
    \label{ex:cred}
    Following Example~\ref{ex:oneShot}, 
    a dominating credential-based mechanism~$M_\cre^\textit{ex}$ of~$M_\os^\textit{ex}$ has the player send only her password and OTP in a single message.
    The rest of the decision process of the mechanism remains the same.
    The pressing of ``start'', which is not a credential, is redundant. 
    
    In practice, the username is necessary to authenticate; however, it is not a credential (not a secret.)
    For simplicity, we abstract it away and assume the mechanism knows in advance which account the authentication is for. 
\end{example}

\negspace
\subsubsection{Deterministic mechanisms}
\label{sec:mech_to_bool:bool_are_better:deter}
Next, we show that in one-shot mechanisms, using randomness to determine if a player is the user or not, does not help the mechanism design.

\begin{lemma}
    \label{lem:credsToDet}
    For all one-shot credential-based mechanisms~$M_\cre$ there exists a deterministic credential-based mechanism~$M_\deter$ that dominates~$M_\cre$.
\end{lemma}

Denote by~$M_\cre$ a credential-based authentication mechanism. 
A deterministic mechanism~$M_\deter$ that behaves the same as~$M_\cre$ except that its step function does not use any randomness can be constructed 
by artificially replacing the randomness used in~$\step{M_\cre}{\cdot}$ with a stream of zeroes.
Thus, preventing~$M_\deter$ from using any randomness in its step function and making it deterministic.
Again, the proof of Lemma~\ref{lem:credsToDet} is similar to that of Proposition~\ref{thm:oneShot} and Lemma~\ref{lem:onlyCredsHelp} and is omitted for brevity.

\begin{example}
    \label{ex:deter}
    Following Example~\ref{ex:cred}, a dominating deterministic credential-based mechanism~$M_\deter^\textit{ex}$ requires both the password and the OTP in a single message.
    If a player sends the correct password and OTP, then she wins.
    Otherwise, the other player wins (instead of randomly choosing a winner). 
    The profile of~$M_\deter^\textit{ex}$ is~$\Prof{M_\deter^\textit{ex}} = \{((1,1),(0,0)), ((1,1),(1,0)), ((1,1),(0,1))\}\supset \Prof{M_\cre^\textit{ex}}$.       
    That is, all scenarios where both credentials are available to the user and not to the attacker. 
\end{example}



\snegspace
\subsubsection{Boolean mechanisms}
\label{sec:mech_to_bool:bool_are_better:bool}
We show that every deterministic credential-based mechanism is dominated by a mechanism defined by a monotonic Boolean function.
\begin{definition}[Boolean mechanism]
    \label{def:bool_func_mechanism}
    Given a Boolean function on~$n$ variables~${f:\{0,1\}^n \rightarrow \{0,1\}}$, we define the \emph{Boolean mechanism}~$M_f$ as follows:
    The credential generation function~$\gen{M_f}{\cdot}$ is a secure credential generation function.
    The step function~$\step{M_f}{\msg,i}$ receives a single message~$\msg$ from a player with identifier~$i$ that contains a set of credentials~$c$.
    The mechanism~$M_f$ extracts the credentials' availability vector~$q$ from~$c$. 
    If~$c$ is not a valid set of credentials,~$M_f$ decides the identifier of the other player~$1-i$.
    If~$f(q)=1$,~$M_f$ decides~$i$. 
    Otherwise,~$f(q)=0$ and the mechanism decides~$1-i$.
    The mechanism~$M_f$ is the \emph{Boolean mechanism of~$f$}.
    If~$f$ is monotonic, then~$M_f$ is the \emph{monotonic Boolean mechanism of~$f$}.
\end{definition}

Intuitively, for a deterministic credential-based mechanism, the only information that matters is whether a player can prove she has a set of credentials.
If a player can authenticate to a mechanism, accessing additional credentials will not prohibit her from authenticating to the same mechanism.
\begin{construction}
    \label{con:M_f}
    Denote by~$M_\deter$ a deterministic credential-based mechanism with~$n$ credentials.
    Define~$f:\{0,1\}^n\rightarrow \{0,1\}$ as follows: 
    For all~$q\in\{0,1\}^n$:
    \begin{align*}
    f(q)=   \begin{cases}
        \text{1} &\: q \in \{\sigma_U \in\{0,1\}| \exists \sigma=(\sigma_U,\sigma_A) \in \Prof{M_\deter}\}\\
        \text{0} &\:\text{otherwise} 
      \end{cases}
    \end{align*} 
    $f$ is the function derived from~$M_\deter$.
    $M_{f}$ is the mechanism defined by~$f$, with the function~$\gen{M_f}{\cdot} := \gen{M_\deter}{\cdot}$. 
    It is easy to see that~$f$ is well-defined. 
    It remains to show that it is monotonic and that it dominates~$M_\deter$.
\end{construction}

\begin{lemma}
    \label{lem:deterToBoolCreds}
    \lemmaDeterToBool
\end{lemma}

\begin{proof}
    Given two binary vectors~${x},{y}\in \{0,1\}^n$, we denote~${x}\geq {y}$ if for every index~$i\in[n]$,~${x}_i\geq {y}_i$. 
    In other words, if an element is set to~$1$ in~${y}$, it is also set to~$1$ in~${x}$.

    We now prove that the function~$f$ is monotonic.
    Let~${x},{y}\in \{0,1\}^n$ such that~${y}\geq{x}$ and~$f(x)=1$.
    By definition of~$f$,~$f(x)=1$ if and only if there exists a scenario~$\sigma$ in which the user's availability vector is~$x$ and~$M_\deter$ succeeds in~$\sigma$. 
    Let~$S_U$ be the winning user strategy for~$M_\deter$ in~$\sigma$.

    As~$M_\deter$ is a credential-based mechanism,~$S_U$ sends only a set of credentials~$c_U$ in a single message.
    Now consider the scenario~$\sigma'$ in which the user's availability vector is~$y$ and the attacker's availability vector is the same as in~$\sigma$.
    The user can use the exact same strategy~$S_U$ in~$\sigma'$, as she has access to all credentials she has in~$\sigma$ and more.
    Thus, for all attacker strategies~$S_A$, schedulers~$\gamma$, and random tapes~$v$, the user wins the execution~$\exec = (M_\deter, \sigma', S_U, S_A, \gamma, v)$. 
    Therefore,~$M_\deter$ succeeds in~$\sigma'$ as well,~${f(y)=1}$ and~$f$ is monotonic.  
\end{proof}

To prove Lemma~\ref{lem:deterToBoolCreds}, we show that if~$M_\deter$ is successful in a scenario~$\sigma$, then~$M_f$ is successful in~$\sigma$ as well.
We build on the observation that based on the scheduler, an execution of~$M_f$ has 3 possible paths: either (1)~$M_f$ receives no messages, thus it does not decide.
Or~(2) the user's message arrives first, with her set of credentials corresponding to the availability vector~$\sigma_U$.
As~$M_\deter$ is successful in~$\sigma$, by definition of~$f$ we get that~$f(\sigma_U) =1$ and~$M_f$ decides the user.

The last possibility is that~(3) the attacker's message arrives first, if it is not a valid set of credentials, then~$M_f$ decides the user.
Otherwise, if it is a valid set of credentials, we show that it is not possible that~$f$ evaluates to~$1$ on the attacker's availability vector.
The proof is somewhat similar to that of Lemma~\ref{lem:M_os_attacker_cant_win} and is done by contradiction.
If~$f$ evaluates to~$1$ on the attacker's availability vector, 
then there exists two execution prefixes of~$M_\deter$ that are indistinguishable to the mechanism~$M_\deter$ and lead to the attacker winning.
In one of which the attacker wins, although the user uses her winning strategy, contradicting the fact that~$M_\deter$ is successful in~$\sigma$.
Thus,~$M_f$ decides the user. 
As in all cases, if a message arrives to~$M_f$, then~$M_f$ decides the user.
We conclude that~$M_{f}$, is a monotonic Boolean mechanism that dominates~$M_{\deter}$.
The proof can be found in Appendix~\ref{app:deterToBoolCreds}.

\begin{example}
    \label{ex:bool}
    Consider the deterministic credential-based mechanism~$M_\deter^\textit{ex}$ from Example~\ref{ex:deter}.
    A dominating monotonic Boolean mechanism~$M_f^\textit{ex}$ is the mechanism defined by the function $f(c_1,c_2) = c_1 \wedge c_2$.        
\end{example}

\snegspace
\subsubsection{Any mechanism is dominated by a Boolean one}
We now show that every mechanism is dominated by a monotonic Boolean mechanism.
\begin{theorem}
    \label{thm:M_to_inc}
    For all profiles~$\prof$, for all mechanisms~$M$ that solve the~$\prof$-asynchronous authentication problem,
    there exists a monotonic Boolean mechanism~$M_f$ that dominates~$M$. 
\end{theorem}

\begin{proof}
    Let~$\prof$ be a profile and let~$M$ be a mechanism that solves the~$\prof$ asynchronous authentication problem.
    By Proposition~\ref{thm:oneShot} there exists a one-shot mechanism~$M_\os$ dominating~$M$, 
    from Lemma~\ref{lem:onlyCredsHelp} there exists a credential-based mechanism~$M_\cre$ dominating~$M_\os$,
    and by Lemma~\ref{lem:credsToDet},~$M_\cre$ is dominated by a deterministic credential-based mechanism~$M_{\textit{det}}$.
    Finally, by Lemma~\ref{lem:deterToBoolCreds} there exists a monotonic Boolean mechanism that dominates~$M_\deter$.
    Overall we get that any mechanism is dominated by a monotonic Boolean mechanism.
\end{proof}

Note that it is not true that for every mechanism there~exists an equivalent Boolean mechanism.
We show an example of a mechanism that has no equivalent Boolean mechanism.

\begin{example}
    \label{exp:no_deter_creds}
    Consider the mechanism~$M_\os^\textit{ex}$ from Example~\ref{ex:oneShot}.
    We showed that,~$\Prof{M_\os^\textit{ex}} = \{((1,1),(0,0))\}$.
    However, there exists no Boolean mechanism with the same profile.
    There are six monotonic Boolean functions with two variables.
    Two of which are the constant functions~$f_1(v) = 0$ and~$f_2(v)=1$ for all~$v\in \{0,1\}^2$ and one can easily confirm that their profiles are empty.
    The other four functions are~$f_3(c_1,c_2) = c_1$,~$f_4(c_1,c_2) = c_2$,~$f_5(c_1,c_2) = c_1 \vee c_2$, and~$f_6(c_1,c_2) = c_1 \wedge c_2$.
    None of these functions has a profile that includes only a single scenario.
\end{example}

\negspace
\subsection{Monotonic Boolean Mechanisms are Maximal}
\label{sec:mech_to_bool:maximal_mechs}

We now prove that every non-trivial monotonic Boolean~mechanism is maximal.
First, we show that a Boolean mechanism is successful in a scenario~$(\sigma_U,\sigma_A)$ if and only~if~the~user~has sufficient credentials and the attacker does not, 
that is,~$f(\sigma_U)=1$ and~$f(\sigma_A)=0$. 
The proof can be found in Appendix~\ref{app:proof_success_equivalence}. 

\begin{observation} \label{cor:profile_of_M_f}
    Let~$M_f$ be a monotonic Boolean mechanism of the function~$f$.
    Let~$T= \{v\in \{0,1\}^n \mid f(v) = 1\}$ and let~$F= \{ v \in \{0,1\}^n \mid f(v) = 0\}$.
    The profile of~$M_f$ is the Cartesian product of~$T$ and~$F$.
    $$\Prof{M_f} = \{(\sigma_U, \sigma_A) \in \{0,1\}^n \times \{0,1\}^n \mid \sigma_U \in T  , \sigma_A \in F \} . $$
    And the profile's size is $|\Prof{M_f}| =|T| \cdot |F|$. 
\end{observation}

We intend to gradually build the mechanism's function. 
To this end, we recall the definition of a partially-defined Boolean function---a function that defines output values for a subset of the Boolean vectors.
\begin{definition}[Partially-defined Boolean function] \cite{partiall_def_bool_func1} 
    \label{def:partially_defined_Boolean_function}
    A \emph{partially-defined Boolean function} (or \emph{partial Boolean function}) is a pair of disjoint sets~$(T,F)$ of binary $n$-vectors. 
    The set~$T$ denotes the set of true vectors and~$F$ denotes the set of false vectors. 
    A Boolean function ${f :\{0,1\}^n \rightarrow \{0,1\}}$ is called an \emph{extension} of~$(T, F)$ if~$f(x) = 1$ for all~$x\in T$ and ~$f(y) = 0$ for all~$y\in F$.     
\end{definition}

We extend the definition of a Boolean mechanism to a partial Boolean mechanism.

\begin{definition}[Partial Boolean mechanism]
    \label{def:mechanism_defined_by_partially_defined_Boolean_function}
    Given a partial Boolean function on~$n$ variables~$(T,F)$, we define the mechanism~$M_{(T,F)}$ as follows:
    The credential generation function~$\gen{M_{(T,F)}}{\cdot}$ is a secure credential generation function.
    The step function~$\step{M_{(T,F)}}{\msg, i}$ receives a single message~$\msg$ from a player with identifier~$i$ that contains a set of credentials~$c$.
    The mechanism~$M_{(T,F)}$ extracts the credentials' availability vector~$q$ from~$c$.
    If~$c$ is not a valid set of credentials,~$M_{(T,F)}$ decides the identifier of the other player~$1-i$.
    If~$q\in T$, the mechanism~$M_{(T,F)}$ decides~$i$. 
    And if~$q \in F$, the mechanism decides~$1-i$.
    Otherwise,~$q\notin T\cup F$, and the step function of~$M_{(T,F)}$ does nothing. 
    $M_{(T,F)}$ is the \emph{partial Boolean mechanism of~$(T,F)$}.
\end{definition}

\begin{note}
    \label{note:profile_of_M_f}
    Similar to Observation~\ref{cor:profile_of_M_f}, the profile of a partial Boolean mechanism of~$(T,F)$ is the set of all scenarios in which the user's availability vector is in~$T$ and the attacker's is in~$F$.
\end{note}
We define a monotonic partial Boolean function as follows:

\begin{definition}[Monotonic partial Boolean function] 
    \label{def:monotonic_partial_Boolean_function}
    A \emph{monotonic partial Boolean function} is a partial Boolean function that has a monotonic extension. 
\end{definition}

For example, in Example~\ref{ex:oneShot}, the mechanism~$M_\os^\textit{ex}$ is the mechanism of the tuple~$(T,F)$ where~$T= \{(1,1)\}$ and~$F= \{(0,0)\}$.
This is a monotonic partial Boolean function, as the function~$f(x,y)=x \wedge y$ is a monotonic extension of~$(T,F)$.

To prove we can indeed extend a partial Boolean function, we make use of the following observation.
We show that if a mechanism~$M_1$ dominates~$M_2$ then the set of user availability vectors in~$\Prof{M_1}$ is a subset of the set of user availability vectors in~$\Prof{M_2}$, and the same holds for the attacker availability vectors.

\begin{observation}
    \label{clm:partial_Boolean_function}
    \partialBooleanfunction
\end{observation}

For example, the mechanisms~$M_1=M_\os^\textit{ex}$ from Example~\ref{ex:oneShot}, and the mechanism~$M_2$ that requires only the first credential to authenticate, both have non-empty profiles, and~$\Prof{M_1} \subseteq \Prof{M_2}$.
Then~$T_1=\{(1,1)\}$,~$T_2=\{(1,1),(1,0)\}$,~$F_1=\{(0,0)\}$ and~${F_2=\{(0,0),(0,1)\}}$.
And it holds that~$T_1 \subseteq T_2$ and~$F_1 \subseteq F_2$.

    

We now show that every mechanism is equivalent to a monotonic partial Boolean one.

\begin{lemma} 
    \label{lem:profile_partial_func}
    \profilePartialFnc
\end{lemma}
Both proofs of Observation~\ref{clm:partial_Boolean_function} and Lemma~\ref{lem:profile_partial_func} can be found in Appendix~\ref{app:proof_partial_Boolean_function}.
Now we show that every maximal mechanism is equivalent to a Boolean mechanism.

\begin{theorem}
    \label{thm:monotonic_maximal}
    For all monotonic Boolean mechanisms of non-constant functions, there exists no strictly dominating mechanism.
\end{theorem}

\begin{proof}
    Let~$f$ be a non-constant monotonic Boolean function,~$M_f$ the mechanism defined by~$f$, and~$M'$ a mechanism dominating~$M_f$.
    We show that~$\Prof{M_f} = \Prof{M'}$.
    
    By Lemma~\ref{lem:profile_partial_func}, there exists a monotonic partial Boolean mechanism equivalent to~$M'$.
    Let~$(T',F')$ be the partial Boolean function defining~$M'$.
    And let~$T$ be the set of all user availability vectors in~$\Prof{M_f}$ and~$F$ be the set of all attacker availability vectors in~$\Prof{M_f}$.
    
    First, we show that the profile of~$M_f$ is not empty.
    As~$f$ is non-constant, there exists a vector~$q$ such that~$f(q)=1$ and a vector~$q'$ such that~$f(q')=0$.
    Consider the scenario~$\sigma$ such that~$\sigma_U=q$ and~$\sigma_A=q'$.
    Then, from Observation~\ref{cor:profile_of_M_f},~${\sigma\in \Prof{M_f}}$, and thus~$\Prof{M_f} \neq \emptyset$.

    Then, as~$\Prof{M_f} \subseteq \Prof{M'}$, from~Observation~\ref{clm:partial_Boolean_function}, we get that~${T\subseteq T'}$ and~$F\subseteq F'$. 
    However, as~$f$ is a Boolean function,~${T\cup F = \{0,1\}^n}$ and~$T\cap F = \emptyset$.
    Therefore,~$T=T'$ and~$F=F'$.
    And thus,~$\Prof{M_f} = \Prof{M'}$.
\end{proof}

Finally, we prove that there exists no hierarchy between monotonic Boolean mechanisms of non-constant functions.
\begin{lemma}
    \label{Lem:different_funcs_different_profiles}
    Let~$f,g:\{0,1\}^n \rightarrow \{0,1\}$ be two different non-constant monotonic Boolean functions. 
    Denote by~$M_f$ and~$M_g$ the monotonic Boolean mechanisms of~$f$ and~$g$ respectively.
    If~$f$ or~$g$ is not constant, then~$\Prof{M_f} \neq \Prof{M_g}$.
\end{lemma}

\begin{proof}
    Let~$f,g:\{0,1\}^n \rightarrow \{0,1\}$ be two different monotonic Boolean functions. 
    There exists a vector~$v \in \{0,1\}^n$ such that~$f(v) \neq g(v)$.
    Assume without loss of generality that~$f(v) = 1$ and~$g(v) = 0$.
    As both functions are non-constant, there exists a vector~$u \in \{0,1\}^n$ such that~${f(u) = 0}$.
    Consider the scenario~$\sigma = (v,u)$.
    We have~$f(v)=1$ and~$f(u)=0$, and thus~$\sigma \in \Prof{M_f}$.
    We also have~$g(v)=0$, and thus~$\sigma \notin \Prof{M_g}$.
    Therefore,~${\Prof{M_f} \neq \Prof{M_g}}$.
\end{proof}

Overall, we showed that for every mechanism that solves the~$\prof$ asynchronous authentication problem, there exists a partial Boolean function that defines it.
And there exists a monotonic Boolean mechanism that dominates it.
In addition, every monotonic Boolean mechanism of a non-constant function is maximal.
Thus concluding that every two monotonic Boolean mechanisms of non-constant functions are either equivalent or incomparable.
This shows that in contrast to the synchronous model~\cite{cryptoeprint:2022/1682}, there exists no probability-agnostic hierarchy between maximal mechanisms in the asynchronous model.

\section{Probabilistic Analysis}
\label{sec:prob}
Using the properties we found of maximal mechanisms, we present an efficient method for finding mechanisms~with approximately maximal success probability.
\subsection{Preliminaries}
As in previous work~\cite{walleddesign}, we derive the probability of scenarios from the credentials' fault probability.
Each credential~$c_i$ can be in one of the four states, 
(1) \emph{safe}: Only~available~to~the user, 
(2) \emph{loss}: Not available to either player, 
(3) \emph{leak}: Available to both players, or 
(4) \emph{theft}: Available only to~the~attacker. 
Its~$c_i^S$ is available to each player accordingly by the scheduler.

Each credential has a probability of being in each of the four states.
For example, a complex password that the user memorized has a low probability of being leaked (guessed) but a high probability of being lost (forgotten).
While a simple password has a high probability of being leaked but a low probability of being lost.
The states of the different credentials are determined by a probability space specified by independent probabilities of each of the credentials. 
The probabilities of a credential~$c_i$ being in each of the four states are denoted by~$P_i^{\safe}$,~$P_i^{\loss}$,~$P_i^{\leak}$, and~$P_i^{\theft}$, respectively, $P_i^{\safe} + P_i^{\loss}+P_i^{\leak} + P_i^{\theft} = 1$.
We assume the mechanism designer estimates the fault probabilities of the credential. 
But note that this is always the case, even if done implicitly.

The probability of a scenario given the probability vector tuple of each credential is the product of the probabilities of each credential being in the state it is in the scenario.
Let~$M_f$ be the Boolean mechanism of the function~$f$. 
The success probability of~$M_f$ is the sum of the probabilities of all scenarios in which~$M_f$ is successful, denoted by~$P^\suc(M_f)$,
and the failure probability is its complement~$P^\textit{fail}(M_f) = 1 - P^\suc(M_f)$. 
A mechanism is \emph{better} than another if its success probability is higher.

The relation between credentials' fault probabilities and the success probability of a mechanism is not straightforward.
For example, consider two credentials~$c_1$ and~$c_2$.
Assume both credentials are prone to loss but not to leakage, with fault probabilities of~$P_1^{\loss} = P_2^{\loss} = 0.1$ and~$P_1^{\safe} = P_2^{\safe} = 0.9$.
Then using both credentials in an OR mechanism, where either credential is sufficient for authentication, results in a success probability of~$0.99$.
While using a single credential results in a success probability of~$0.9$. 
Therefore, in this case, using both credentials increases the success probability of the mechanism.
However, this is not true in general.
For example, assume one credential is prone to loss with a fault probability of~$P_1^{\loss} = 0.1$ and~$P_1^{\safe} = 0.9$ and the other is prone to leakage with~$P_2^{\leak} = 0.1$ and~$P_2^{\safe}= 0.9$.
Then using both credentials in an OR mechanism results in a success probability of~$0.9$.
Moreover, any Boolean mechanism using one or both credentials has the same success probability of~$0.9$.
Therefore, in this case, using two credentials does not increase the success probability of the mechanism compared to using a single credential.

\subsection{Profiles and Scenarios}
Given the number of credentials~$n$ and the probabilities of states, our goal is to find a mechanism with the highest success probability.
Using exhaustive search over the space of all possible mechanisms (equivalently, all monotonic Boolean functions) is infeasible for more than just a few credentials~\cite{walleddesign}.
The number of different Boolean functions with~$n$ credentials is the Dedekind number of~$n$~\cite{Dedekind1897} which grows super-exponentially and is only known up to~$n=9$~\cite{vanhirtum2023computation}.
Therefore, we aim to choose a mechanism that is close to optimal while exploring only a small fraction of the space of all mechanisms.
We achieve this by searching by scenarios.

First, we show that there is no need to consider scenarios with no safe credentials.
\begin{definition}[Viable scenarios]
    A \emph{viable scenario}~$\sigma$ is a scenario in which there exists at least one safe credential. 
    A scenario is~\emph{non-viable} if it is not viable.
    That is,~$\sigma_U \leq \sigma_A$.
\end{definition}

Maram et al.~\cite{cryptoeprint:2022/1682} showed that in the synchronous model, a non-viable scenario is not in the profile of any mechanism.
As their model is stronger (has additional assumption on the network), this is also true in our asynchronous model. 
 We also use their following observation. 



\begin{observation}
    \label{obs:non_viable_scenarios_num}
    \numViableScenarios
\end{observation}


Now we bound the number of scenarios that can be added to any partial Boolean mechanism's profile.
\begin{observation}
    \label{obs:scenarios_addition}
    \profileSizeBound
\end{observation}

Both proofs of Observation~\ref{obs:non_viable_scenarios_num} and Observation~\ref{obs:scenarios_addition} are in Appendix~\ref{app:profile_size_bound}.

    
    
    

\begin{figure}[t]
    
    \centering
    \begin{subfigure}{0.49\columnwidth}
        \centering
        \includegraphics[width=\linewidth]{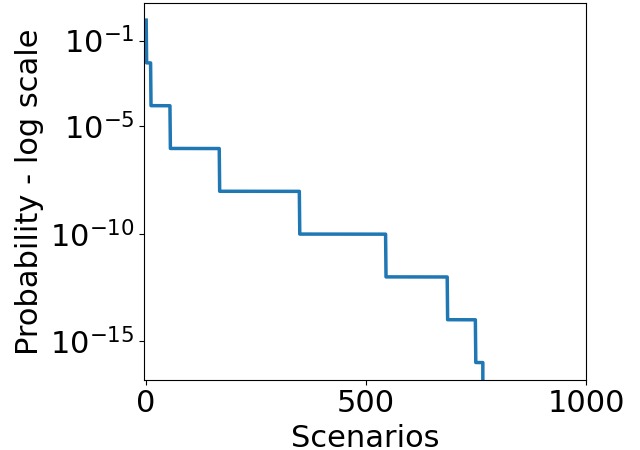} 
        \caption{$n=9$: ${P_1^{\loss}=P_1^{\leak}=0.01}$, for~$i>1$~$P_i^{\theft} = 0.01$.}
        \label{fig:probability_9_keys_hetro}
    \end{subfigure}\hfill
    \begin{subfigure}{0.49\columnwidth}
        \centering
        \includegraphics[width=\linewidth]{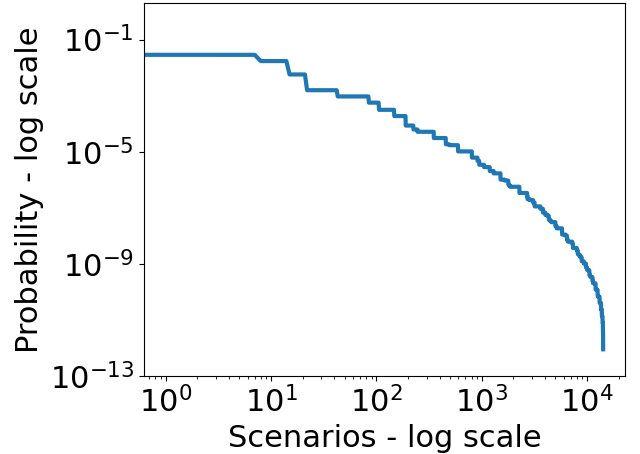} 
        \caption{$n=7$ with~${P^\loss = 0.05}$, $P^\leak=0.03, P^{\theft} = 0.01$.}
        \label{fig:probability_7_keys_hetro}
    \end{subfigure}

\caption{Probability of scenarios for~$n$ credentials and different probabilities.}
\label{fig:prob}
\end{figure}

To maximize success probability, we limit our search to mechanisms defined by monotonic Boolean functions.
We are only interested in credentials with low fault probabilities. 
Intuitively, scenarios with fewer faults have higher probabilities than those with more faults. 
So the probability distribution of viable scenarios is highly concentrated on a small number of scenarios.
And when some fault probabilities are~$0$, the number of scenarios with non-zero probability drops even further. 
Thus, the number of scenarios that our algorithm needs to consider is much smaller than the total number of scenarios.
For example, Figure~\ref{fig:probability_9_keys_hetro} shows the probability distribution of the scenarios for a system with 9 credentials where the first credential might be lost or leaked with probability~$P_i^{\loss} = P_i^{\leak} = 0.01$ and the rest of the credentials might be stolen with probability~$P_i^{\theft} = 0.01$.
The number of scenarios with~$9$ credentials is~$4^9=$ \num[group-separator={,}]{262144} with~\num[group-separator={,}]{242461} viable scenarios, but only~\num{766} have non-zero probability.
Similarly, Figure~\ref{fig:probability_7_keys_hetro} shows the probability distribution of the scenarios for a system with 7 credentials where all fault types are possible with probability~$P_i^{\loss}= 0.05, P_i^{\leak} = 0.03$, and~$P_i^{\theft} = 0.01$ for all credentials.
The number of different scenarios with 7 credentials is~$4^7 = $ \num[group-separator={,}]{16384}, out of which~\num[group-separator={,}]{14197} are viable scenarios with positive probability.
Yet, over~$91\%$ of the probability is concentrated in the first~$50$ scenarios and over~$95\%$ on the first~$100$ scenarios.
Other fault probabilities result in qualitatively similar results.
While this case demonstrates a more challenging situation, it still shows that the significant part of the probability is concentrated on a small number of scenarios.
Thus, when searching for near-optimal mechanisms, we can neglect the vast majority of scenarios. 


\snegspace

\begin{algorithm}[t]
    \SetInd{0.1em}{0.9em}
    \SetAlgoNoLine 
    \SetAlgoNoEnd 
    \DontPrintSemicolon 
        
    \SetKwInOut{Input}{input}\SetKwInOut{Output}{output}
    
    \nonl \emph{\viableScenarios}:  List of all viable scenarios, sorted by probability in descending order, calculated based on the fault probabilities of the credentials. \;
    \nonl $\maxSuccessProb \gets 0$ \tcp*[r]{maximal success probability found so far}
    \nonl $\maxTable \gets$ truth table where the all zeroes row is set to~$0$, the all ones row is set to~$1$, and the rest is set to~$\bot$ \tcp*[r]{truth table with the maximal success probability found so far}
    \nonl $\delta$: The precision parameter chosen by the user.

    \nonl \small \emph{\scenarioBasedSearch(\truthTable, $\indexx$)} \scriptsize \; 

    \tcc{Input: \truthTable - a (possibly partial) monotonic truth table of a mechanism.}
    \tcc{Input: $\indexx$ is the index of the next scenario to consider.}
    \tcc{Output: No returned value. The function updates~$\maxSuccessProb$ and~$\maxTable$.}

    $\currProfile\gets$ \truthTable's profile, calculated as described in Note~\ref{note:profile_of_M_f} \label{alg:line:get_profile}\; 
    $\successProb \gets$ success probability of \currProfile \label{alg:line:curr_success}\;  
    
    \If{\truthTable is complete}  {  \label{alg:line:complete_table} 
        \If{\successProb > \maxSuccessProb} { \label{alg:line:success_prob_condition}
            $\maxSuccessProb \gets \successProb$ \label{alg:line:max_update_1}\;
            $\maxTable \gets \truthTable$ \label{alg:line:max_update_1_2}\;
        }
        \Return \;
    }
    $\possibleScenariosProbabilitiesSorted \gets$ sorted array of scenarios probability that can be added to truthTable and whose index in viableScenarios is higher than~\indexx  \;  \label{alg:line:sorted_additional_probabilities}

    $\numPossibleAdditions \gets$ maximum number of scenarios that can be added to \truthTable given the table as calculated in Observation~\ref{obs:scenarios_addition} \label{alg:line:num_possible_additions}\; 
    $\maxPossibleAdditions \gets \sum_{i=1}^{\numPossibleAdditions} {\possibleScenariosProbabilitiesSorted(i)}$  \label{alg:line:max_possible_additions} \;
    \If(\tcp*[f]{additional scenarios will add 0 to the success probability}){$\maxPossibleAdditions = 0$}  { \label{alg:line:zero_additions}
        $\truthTable \gets$ arbitrarily complete \truthTable \label{alg:line:complete_table_arbitrarly} \;
        \If{$\successProb > \maxSuccessProb$} { \label{alg:line:arbitrary_better}
            $\maxSuccessProb \gets \successProb$ \label{alg:line:max_update_2}\;
            $\maxTable \gets \truthTable$ \label{alg:line:max_update_2_2}\; 
        }    
        \Return \;
    }

    $\potentialSuccessProb \gets \successProb + \maxPossibleAdditions$ \tcp*{upper bound on success probability} 

    \If{$\maxSuccessProb > 0$ \textbf{and} $\maxSuccessProb > \potentialSuccessProb - \delta$}   {  \label{alg:line:delta_condition}
        \Return \label{alg:line:delta_pruning} \tcp*{cannot exceed by over $\delta$, prune this branch}
    }

    $\sigma \gets \viableScenarios(\indexx)$ \label{alg:line:get_next_scenario} \tcp*{get next scenario}
    \prevUserVal = \truthTable($\sigma_U$)  \tcp*{save previous values}
    \prevAttVal =  \truthTable($\sigma_A$) \; 

    \If {($\prevUserVal \neq 0$ \textbf{and} $\prevAttVal \neq 1$) \textbf{and} ($\prevUserVal = \bot$  \textbf{or} $\prevAttVal = \bot$)} {  \label{alg:line:include_scenario_cond}
        $\tableWithScenario \gets$ updateTable(\truthTable, $\sigma$) \label{alg:line:update_table} \tcp*{update table with $\sigma$ and keep it monotonic} 

        \scenarioBasedSearch(\tableWithScenario, $\indexx+1$) \label{alg:line:include_scenario} \tcp*{recursive call with the new scenario}
    }
    \scenarioBasedSearch(\truthTable, $\indexx+1$)    \label{alg:line:skip_scenario} \tcp*{recursive call without the new scenario}

    \caption{Scenario-based search algorithm}
\label{alg:scenario_based_search}
\end{algorithm}

\snegspace
\subsection{Scenario-Based Search Algorithm}
Given a positive fraction~$\delta\in \mathcal{N}^+$, 
our scenario-based search algorithm (Algorithm~\ref{alg:scenario_based_search}) finds a mechanism whose success probability is at most~$\delta$ below that of an optimal mechanism.
It takes advantage of three key observations: 
(1) The number of scenarios that might be added to a profile of a partial Boolean mechanism is bounded~(Observation~\ref{obs:scenarios_addition}); 
(2) non-viable scenarios can be ignored~\cite{cryptoeprint:2022/1682};
and (3) the probability of different scenarios drops quickly~(Figure~\ref{fig:prob}).

Instead of searching all possible mechanisms to find the optimal one, we build the truth table (and thus mechanism) incrementally:
Given the number of credentials~$n$, we start with the empty partial truth table of~$n$ variables of the function (total of~$2^n$ rows) and add rows.
By Lemma~\ref{lem:profile_partial_func}, for all mechanisms~$M$ there exists a monotonic partial Boolean function~$(T,F)$ such that~$M$ is equivalent to the mechanism defined by~$(T,F)$.
The monotonic partial truth table corresponding to~$(T,F)$ has the value~$1$ for all vectors in~$T$,~$0$ for all vectors in~$F$, and the rest of the vectors~$v\notin T\cup F$ are not set yet (denoted by~$\dontCare$).
Changing a value of a row in the truth table from~$\dontCare$ to~$0$ is equivalent to adding the vector to~$F$ and changing a value from~$\dontCare$ to~$1$ is equivalent to adding the vector to~$T$.

By Theorem~\ref{thm:M_to_inc} every mechanism is dominated by a Boolean mechanism.
Because every maximal mechanism is a monotonic Boolean mechanism (Theorem~\ref{thm:monotonic_maximal}), and every monotonic Boolean function yields a maximal mechanism~(Lemma~\ref{Lem:different_funcs_different_profiles}),
we extend the monotonic partial Boolean function that defines the mechanism until it is full (representing a Boolean function).

Note that if the value of the all zeroes row is set to~$1$ in the truth table, from monotonicity, we get the constant function~$f(x)=1$:
As for all~${x\in\{0,1\}^n}$ we have that~$x \geq 0^n$, then to ensure monotonicity, we must set~$f(x)=1$ for all~$x\in\{0,1\}^n$. 
Similarly, if the value of the all ones row is set to~$0$, we get the constant function~$f(x)=0$.

It is easy to confirm that the profile of a Boolean mechanism of a constant function is empty:
Let~$\omga_\F$ and~$\omga_\T$ be the mechanisms defined by the constant Boolean functions~$\forall v, \F(v)=0$, and~$\T(v)=1$.
Let~${\sigma= (\sigma_U,\sigma_A)}$ be a scenario.
Because~${\forall v, \F(v)=0}$, we get that~${\F(\sigma_U)\wedge \neg \F(\sigma_A) = 0}$, thus~$M_\F$ fails.
Similarly, we have~${\T(\sigma_U)\wedge \neg \T(\sigma_A) = 0}$.
Therefore, both \mech s have empty profile.
So we consider only non-constant monotonic Boolean functions and set the value of the all zeroes row to~$0$ and the value of the all ones row to~$1$.

Equipped with these insights, we are ready to present the scenario-based search (Algorithm~\ref{alg:scenario_based_search}). 
Given the credentials' probability vectors, we calculate the probability of each viable scenario, sort it and save it in a global variable \viableScenarios. 
Then we progress greedily, from scenarios with the highest to the lowest probabilities.
For each scenario, there are two options: either include it in the profile or not.
We refer to \truthTable as a dictionary mapping an availability vector to~$0$,~$1$, or~$\dontCare$.
We denote by~$\truthTable(\sigma_U)$ and~$\truthTable(\sigma_A)$ the value of the user and the attacker's availability vectors~$\sigma_U$ and~$\sigma_A$ in~\truthTable respectively.
Similarly, we denote by~${\viableScenarios(\indexx)}$ the scenario at index~\indexx in the sorted list of viable scenarios.

In each step, the algorithm computes the success probability of the current partial truth table (lines~\ref{alg:line:get_profile}-\ref{alg:line:curr_success}).
It then checks if the truth table is full (line~\ref{alg:line:complete_table}). 
If so, it means we reached a Boolean function that cannot be extended.
The algorithm checks if the success probability of the mechanism (truth table) is higher than the best found so far (line~\ref{alg:line:success_prob_condition}).
If so, it updates the best success probability and the corresponding global best truth table (lines~\ref{alg:line:max_update_1}-\ref{alg:line:max_update_1_2}), then it returns. 

Otherwise, the table can be extended.
The algorithm checks if it can add any scenarios with positive probability to the current truth table (line~\ref{alg:line:zero_additions}).
If not, then any additional scenarios would contribute~$0$ to the success probability, regardless of the exact scenarios chosen.  
Thus, it completes the truth table arbitrarily by iteratively adding scenarios from those left (with a higher index than \indexx in $\viableScenarios$) while making sure the result is monotonic (line~\ref{alg:line:complete_table_arbitrarly}),
checks if the success probability of the resulting mechanism  is higher than the global best found so far (line~\ref{alg:line:arbitrary_better}), and updates the global best success probability and the corresponding global best truth table accordingly (lines~\ref{alg:line:max_update_2}-\ref{alg:line:max_update_2_2}), then it returns.
Note that this covers the case where no scenarios are left to add. 

Then, the algorithm checks if it is beneficial to continue extending the current table.
It does so by finding an upper bound on the additional success probability for the current branch by considering the next~$\textit{\numPossibleAdditions}$ scenarios in the sorted list and summing their probabilities (line~\ref{alg:line:max_possible_additions}) where~$\textit{\numPossibleAdditions}$ (line~\ref{alg:line:num_possible_additions}) is the maximum number of scenarios that may be added to the current truth table without contradicting it~(Observation~\ref{obs:scenarios_addition}).
If the potential best solution is not better by at least~$\delta$ than the current best (line~\ref{alg:line:delta_condition}), we prune this branch and do not consider any of the mechanisms it results in (line~\ref{alg:line:delta_pruning}).
This takes advantage of the fact that the probability of different scenarios drops exponentially fast (Figures~\ref{fig:probability_9_keys_hetro} and~\ref{fig:probability_7_keys_hetro}) and is a key step in ensuring the algorithm's efficiency.

If none of the stop conditions hold, the algorithm continues by exploring the next scenario (line~\ref{alg:line:get_next_scenario}).
For the recursive exploration, given a scenario, we have two options: (1) Include the scenario in the truth table and recursively call the algorithm (only if it does not contradict the current truth table and is not already in the profile) (line~\ref{alg:line:include_scenario}). 
Or (2) Exclude this scenario and recursively call the algorithm with the next scenario (line~\ref{alg:line:skip_scenario}).

Given a scenario~$\sigma$ and a partially filled truth table,
we check whether the scenario contradicts the truth table and whether it is already in the profile (line~\ref{alg:line:include_scenario_cond}).
We do so by checking if the value of the user's availability vector~$\sigma_U$ is not set to~$0$ (either~$1$ or $\dontCare$) and the value of the attacker's availability vector~$\sigma_A$ is not set to~$1$ (either~$0$ or $\dontCare$).
If so, the scenario does not contradict the truth table.
And checking that either the user's availability vector or the attacker's availability vector is not set yet (has a value $\dontCare$).
If so, then it is not in the profile.
If both conditions hold, we include the scenario in the truth table and recursively call the algorithm (line~\ref{alg:line:include_scenario}).

To include a scenario~$\sigma$ in the profile, we create an updated truth table using the function updateTable(\truthTable, $\sigma$) (line~\ref{alg:line:update_table}). 
The function sets the entry corresponding to~$\sigma_U$ to~$1$ and the entry~$\sigma_A$ to~$0$.
Then, to have the truth table represent a monotonic Boolean function, updateTable($\cdot$) updates it by setting all the entries~${x\in\{0,1\}^n}$ such that~${x>\sigma_U}$ to~$1$.
And setting all entries~${y\in\{0,1\}}$~such~that~${y < \sigma_A}$~to~$0$. 

The algorithm's result is the mechanism saved in the \maxTable and the corresponding success probability \maxSuccessProb. 


\subsection{Algorithm Correctness}
\label{sec:algo_correctness}
We now show that our scenario-based search~(Algorithm~\ref{alg:scenario_based_search}) finds a mechanism that is~$\delta$ close to the optimal mechanism.
First, we note that after each update to the partial monotonic truth table, it remains a (possibly partial) monotonic truth table---the above update does not contradict previous ones (Appendix~\ref{app:algorithm_correctness}).
Next, we show that the algorithm always stops and that when it returns, the truth table contains a valid complete monotonic truth table whose success probability is at most~$\delta$ less than the optimal mechanism.

The algorithm explores adding scenarios to the truth table incrementally, in the worst case it may explore all viable scenarios.
As the number of viable scenarios is finite and bounded (Observation~\ref{obs:non_viable_scenarios_num}), the algorithm always stops.

The algorithm stops only when one of the three stopping conditions is met.
In the first case (line~\ref{alg:line:complete_table}), the truth table is complete, and the algorithm updates \maxTable to this complete table.
In the second case (line~\ref{alg:line:zero_additions}), the algorithm finds that all scenarios that can be added to the current truth table will not improve the success probability.
So it completes the truth table arbitrarily, and \maxTable gets the completed table.
Finally, in the third case (line~\ref{alg:line:delta_condition}), any possible additions do not result in a mechanism that is at least~$\delta$ better so the algorithm prunes the branch and returns, without updating \maxTable.
However, this condition is met only if \maxSuccessProb is greater than 0 i.e., if the algorithm found a solution earlier.
Therefore, the algorithm must update a full truth table at least once, and it is always a valid monotonic truth table that has a success probability at most~$\delta$ less than the optimal mechanism.
The proof is in Appendix~\ref{app:algorithm_correctness}.

\begin{figure}[t]
    \centering
    \includegraphics[width=0.47\textwidth]{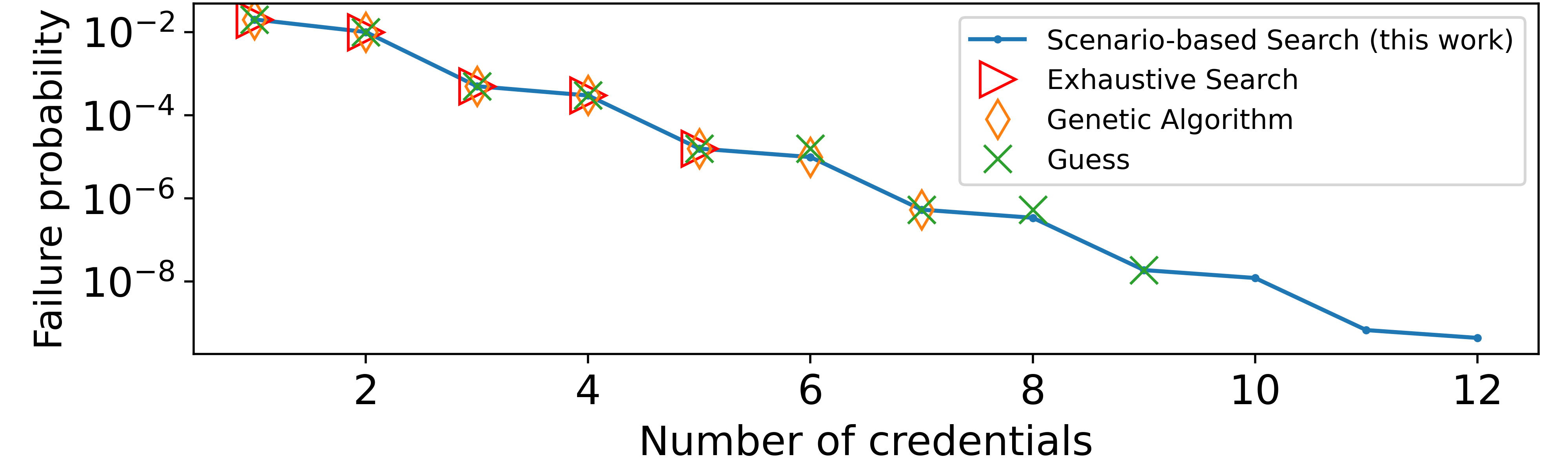}
    \caption{Failure probability vs. credentials number. Comparing our algorithm to previous methods: exhaustive search, genetic algorithm, and guessing a symmetric mechanism.}
    \label{fig:results_hetro}
\end{figure}

\subsection{Algorithm Complexity}
\label{sec:algo_complexity}
We first present an upper bound on the complexity of the algorithm, followed by a discussion of when and why the actual complexity is much lower.
Then we evaluate the complexity of the algorithm empirically~(\S\ref{sec:empirical_complexity}), by measuring the runtime of our algorithm for different numbers of credentials and fault probabilities.

\subsubsection{Bound}
\label{sec:disscusion_complexity}

To give a loose upper bound on the algorithm's time complexity, we note that each recursive call's complexity is bounded by~$O((4^n-3^n)\cdot n)$.
And in the worst case, the recursion depth reaches~$O(2^{4^n-3^n})$.
The overall bound is thus~$O(2^{4^n-3^n}\cdot (4^n-3^n)\cdot n)$.
We defer a detailed analysis to Appendix~\ref{app:sec:disscusion_complexity}.

This analysis assumes the worst case for every step of the algorithm.
However, in practice, many steps' worst case cannot happen simultaneously.
Thus, the actual complexity of the algorithm is much lower, depending  on the number of credentials, their probabilities, and the parameter~$\delta$. 

For example, while an upper bound on the recursion depth is given by 2 to the power of the number of viable scenarios, the actual depth explored is much smaller.
This is due to multiple factors, for example (1) the exponential drop in the probability of different scenarios, combined with the fact that the algorithm prunes branches with negligible advantage;
(2) the fact that the number of scenarios that can be added to a profile of a partial Boolean mechanism is limited~(Observation~\ref{obs:scenarios_addition})
in a non-trivial way depending on which do result in a monotonic mechanism;
and (3)~each scenario adds not a single row to the truth table, but possibly many rows for keeping monotonicity 
(depending on the specific scenario and the current truth table).
While all those factors contribute to the algorithm's efficiency, 
it remains an open question whether the algorithm's exact complexity can be calculated theoretically (cf. the lack of a closed form expression for the number of monotinic Boolean functions~\cite{Dedekind1897}).

\subsubsection{Empirical Complexity}
\label{sec:empirical_complexity}

As the problem of finding the exact complexity of the algorithm remains open, 
we evaluate the algorithm's complexity by measuring its runtime as a function of the number of credentials and their different fault probabilities.
Figure~\ref{fig:runtime} shows the runtime of the algorithm for different numbers of credentials and fault probabilities using a logarithmic Y~scale. 

We consider two different categories of results: One where the algorithm's runtime grows exponentially with the number of credentials, and one where it grows super-exponentially.

\textbf{Exponential Growth:}
When one credential can suffer from up to a single type of fault and the rest can have up to two types of faults,
or when all credentials can suffer from all types of faults with low probability,
the algorithm's runtime grows exponentially ($O(4^n)$) with the number of credentials (all fits with~${R^2 > 0.97}$, exact values are in Appendix~\ref{app:sec:empirical_complexity}). 
For example, when all credentials can suffer from loss with~$P^\loss = 0.01, P^\leak=P^\theft=0$ 
(\emph{only loss} in Figure~\ref{fig:runtime}), 
or when 2 credentials can be easily lost, but not leaked or stolen, with~$i\in\{1,2\}$,~$P_i^\loss=0.3, P_i^\leak=P_i^\theft=0$ and the rest can be either lost or leaked with~$j>2$,~$P_j^\loss=P_j^\leak=0.01,P_j^\theft=0$. 
(\emph{2~easily-to-lose} in Figure~\ref{fig:runtime}).

We observe a similar trend when each credential can have all three types of faults, where at least two with low fault probabilities, $P^\loss=0.01, P^\leak=P^\theft=0.001$
(\emph{loss, leak, theft} in Figure~\ref{fig:runtime}). 

In all such cases, our scenario-based search runtime grows exponentially slower (better) than the exhaustive search which requires at least $O(D(n)\cdot 4^n)$ operations, where $D(n)$ is the number of monotonic Boolean functions of $n$ variables. 
Notice that this improvement was crucial for enabling the forthcoming case studies. 

\textbf{Super-Exponential Growth:}
But in some cases complexity is worse. 
When all credentials can suffer from two or more types of faults with a high probability, the algorithm's runtime grows super-exponentially with the number of credentials.
E.g., if all three fault types have high probabilities, say, $P^\loss=0.1, P^\leak=0.3, P^\theft=0.4$, 
the algorithm's runtime grows too rapidly to obtain sufficient data points for regression analysis. 


\begin{figure}[t]
    \centering
    \includegraphics[width=0.47\textwidth]{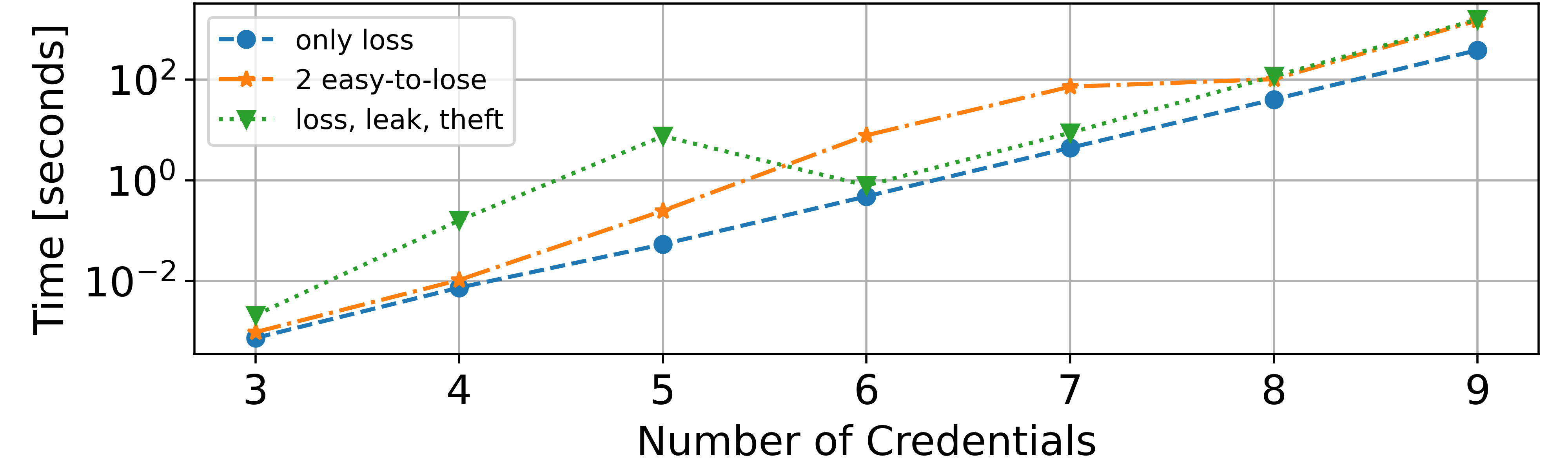}
    \caption{Runtime of the algorithm as a function of the number of credentials for different fault probabilities with $\delta=10^{-5}$.}
    \label{fig:runtime}
\end{figure}

\snegspace

\subsection{Case Studies}
\label{sec:case_studies}
The algorithm allows designers to choose approximately-optimal mechanisms given their credentials' fault probabilities. 
As the actual fault probabilities of credentials are application-specific and not publicly available, a user or a company can use the algorithm to find the best mechanism for their actual case.
The ability to study mechanisms with a larger number of credentials compared to previous work allows taking advantage of elaborate schemes. 
Our focus is on providing a tool for finding approximately-optimal mechanisms rather than on specific mechanisms.
We use concrete values to evaluate the algorithm's efficacy and accuracy~(\S\ref{sec:case_studies:algo_eval}) and propose improvements for cryptocurrency wallets~(\S\ref{sec:case_studies:crypto_wallets}) and for online authentication~(\S\ref{sec:case_studies:sec_questions}).


\snegspace
\subsubsection{Algorithm evaluation:} \label{sec:case_studies:algo_eval}
We demonstrate the efficacy and accuracy of our method by comparing the results to the optimal mechanisms found using exhaustive search, a heuristic genetic algorithm~\cite{walleddesign}, and guessing a symmetric~${k/n}$ mechanism~\cite{walleddesign}. 
As expected, the number of credentials has a major effect on wallet security.
Figure~\ref{fig:results_hetro} shows the security of the optimal mechanism for different numbers of credentials when the first credential can be lost or leaked, but not stolen ($P^\loss = P^\leak = 0.01$), and the rest of the credentials can only be stolen ($P^\theft = 0.01$).
We use a parameter~$\delta$ ranging from~$10^{-5}$ for smaller numbers of credentials to~$10^{-6}$ for larger numbers. 
Our algorithm efficiently finds the same optimal mechanism as the exhaustive search and is able to find an approximately optimal mechanism for up to 12 credentials.
In several cases, it finds better mechanisms than the ones found by guessing a symmetric mechanism, e.g., Figure~\ref{fig:results_hetro} with 6 or 8 credentials.

\snegspace
\subsubsection{Cryptocurrency wallets} \label{sec:case_studies:crypto_wallets}

Typically, cryptocurrency wallets use 
a single private key or mnemonic~\cite{coinbase, Metamask} or multiple such credentials with so-called multisig wallets or using threshold signatures~\cite{sepior, Camino}, e.g., 2 out of 2 or 2 out of 3.
These mechanisms rely on users' ability to securely store their credentials, and thus require low credential fault probabilities.
While increasing the number of strong keys improves security~\cite{walleddesign}, storing strong keys is challenging~\cite{SoK_The_Quest_to_Replace_Passwords, Mental_Models}.
However, it is possibly easier to store credentials with high loss probability, e.g., a hard password committed to memory that the user might forget but is hard to guess.

By exploring the larger design space with many credentials, we find wallets can take advantage of weak credentials to improve security.
For example, in the case of a 2/2 multisig mechanism, 
using two additional weak credentials reduces the failure probability by an order of magnitude from~$2\cdot 10^{-2}$ to~$6\cdot 10^{-3}$.
Figure~\ref{fig:easy_to_lose} shows the failure probability for different numbers of credentials when all credentials can be lost or leaked, but not stolen ($P^\loss = P^\leak = 0.01$), compared to the security when using the same credentials but adding weak credentials that can only be lost with higher probability ($P^\loss = 0.3$).
Here we used the parameter~$\delta = 10^{-6}$.
To the best of our knowledge, there are no public user studies or statistics on credential fault probabilities.
However, we considered different combinations of possible fault probabilities.
The results are qualitatively similar.  
We find that for~$n$ regular credentials and~$k$ easy-to-lose credentials, the optimal mechanism design is to require any of the~$k$ easy-to-lose credentials or any~$\lfloor n/2\rfloor +1$ of the~$n$ regular credentials.

\begin{figure}[t]
    \centering
    \includegraphics[width=0.47\textwidth]{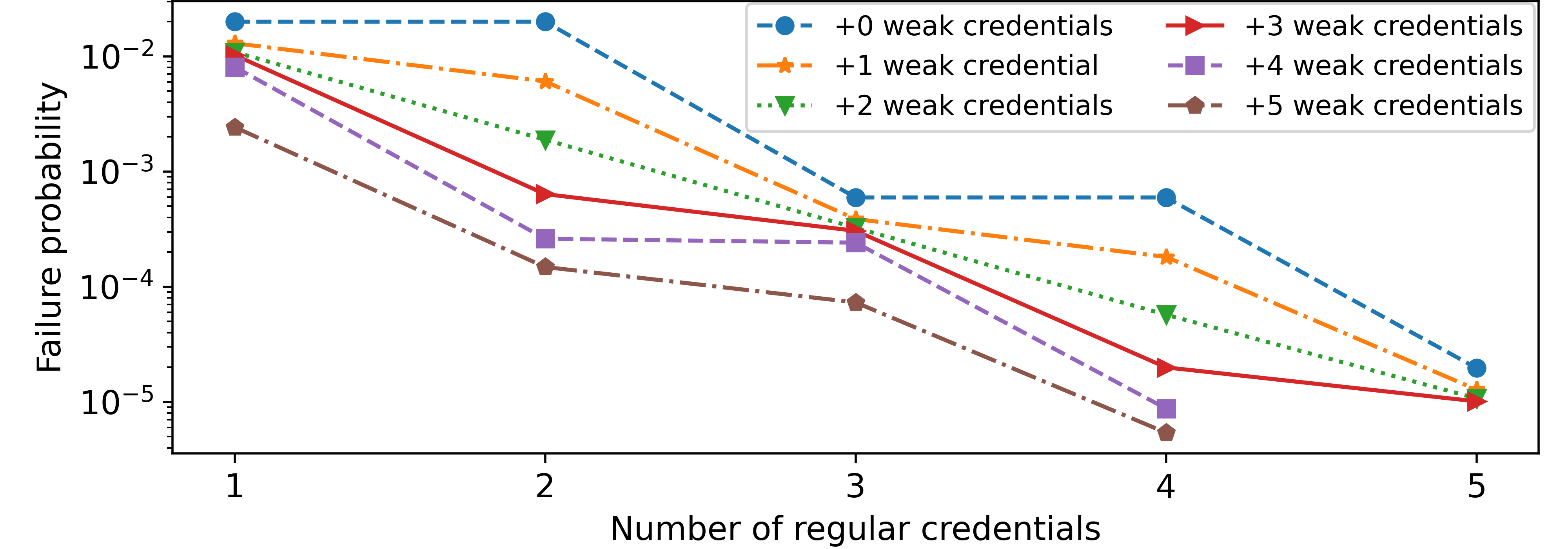}
    \caption{Failure probability harnessing easy-to-lose credentials. Weak:~$P^\leak=0, P^\loss =0.3$, regular:~$P^\leak = P^\loss=0.01$.}
    \label{fig:easy_to_lose}
\end{figure}

\negspace
\snegspace
\subsubsection{Security questions} \label{sec:case_studies:sec_questions}
Until a few years ago, large email providers like Google, Microsoft, and Yahoo!\ employed security questions (like ``name of first pet'') for password reset in online services~\cite{SoK_The_Quest_to_Replace_Passwords}. 
But using security questions in this simple manner famously reduces security~\cite{schechter2009s, rabkin2008personal}. 
The monotonic Boolean function representing such a mechanism is: 
require all answers to the security questions or the password.
Or in general, require all answers to the security questions or any~$k$ out of~$n$ of the regular credentials (e.g., password, fingerprint, or SMS code). 
We call this the classical mechanism.

The main issue is that the answers to the security questions are often easy to guess or to find online, that is, they are easy-to-leak credentials.
This implies that, with a fairly high probability, an attacker can gain access to the account by answering the security questions. 
And, indeed, most online services have stopped using security questions~\cite{bonneau2010password} (with a few exceptions, e.g.,~\cite{leumi, mtbank}). 

Perhaps surprisingly, we find that a well-designed mechanism can actually take advantage of easy-to-leak credentials. 
We used the algorithm to find an approximately optimal mechanism and compared it against the classical mechanism. 
Figure~\ref{fig:easy_to_leak} shows the failure probability of the mechanism for different numbers of credentials when all credentials can be lost or leaked ($P^\loss = P^\leak = 0.01$), compared to the security when using the same credentials but adding weak credentials that can only be leaked with higher probability ($P^\leak = 0.3$).
Just two easy-to-leak keys reduce the failure probability by about an order of magnitude.
Again, we experimented with different combinations of possible fault probabilities and the results were qualitatively similar.  
The resultant mechanism requires answering the security questions in addition to having any~$\lceil n/2 \rceil$ out of the~$n$ regular credentials. 



\begin{figure}[t]
    \centering
    \includegraphics[width=0.47\textwidth]{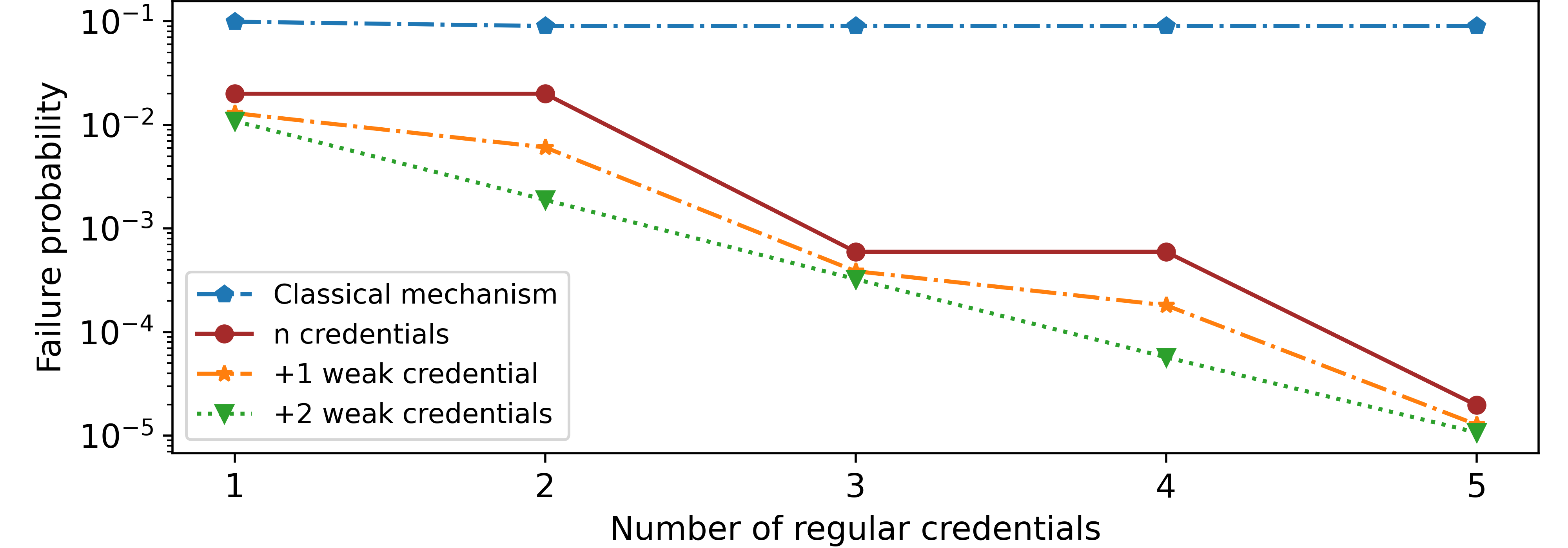}
    \caption{Failure probability harnessing easy-to-leak credentials. Weak:~$P^\leak=0.3, P^\loss =0 $, regular:~$P^\leak = P^\loss=0.01$.}
    \label{fig:easy_to_leak}
\end{figure}


\section{Conclusion}
\label{sec:conclusion}

We formalize the common asynchronous authentication problem, {reconciling the apparent tension between asynchrony and cryptographic security.} 
We show every mechanism is dominated by a non-trivial monotonic Boolean mechanism, which itself is not strictly dominated. 
So in every pair of distinct such mechanisms, neither strictly dominates the other.
We present a practical algorithm for finding approximately optimal mechanisms, given the credential fault probabilities.
This highlights the need for user studies to quantify those probabilities. 


Our algorithm reveals surprising results; Accurately incorporating weak credentials improves mechanisms' security by orders of magnitude. 
One case study shows that a user designing her cryptocurrency wallet can benefit from incorporating with her carefully-guarded credentials a few additional ones that are perhaps easy to lose. 
This is achieved by a mechanism that requires any of the easy-to-lose credentials or any~$\lfloor n/2\rfloor +1 $ out of the~$n$ regular credentials. 
Another case study shows that using security questions (that are easy-to-leak) in addition to regular credentials can significantly improve security if used wisely, by requiring any~$\lceil n/2 \rceil$ out of the~$n$ regular credentials and the answers to the security questions.
More generally, both individuals and companies can use these results directly to rigorously design authentication mechanisms addressing the demands of current and forthcoming challenges.

\section*{Acknowledgments}
This work was supported in part by Avalanche Foundation and by IC3.


\snegspace

\bibliographystyle{ACM-Reference-Format}
\bibliography{main}

\appendix

\section{Security Parameters and Mechanisms Families}
\label{app:full_model}
\subsection{Mechanisms With Security Parameters}
\label{app:mechanisms_with_security_parameters}
\begin{definition}[Authentication mechanism]
    An \emph{authentication mechanism}~$M^\lambda$ is a finite deterministic automaton, parametrized by the security parameter~$\lambda$ and specified by two functions,~$\gen{M^\lambda}{\cdot}$ and~$\step{M^\lambda}{\cdot}$: 
    \label{app:Authentication mechanism}
    \begin{itemize}
        \item $\gen{M^\lambda}{n}$, where~$n$ is the number of credentials (Note that the security parameter is given implicitly with~$M^\lambda$), is a secure credential generator.  And
        \item $\step{M^\lambda}{\msg, i}$, where~$\msg$ is the message the mechanism received (perhaps no message, represented by~${\msg=\bot}$), sent from player~$i$, is a function that updates the state of the mechanism and returns a pair~$(\msg s_0,\msg s_1)$ of sets of messages (maybe empty) to send to the players by their identifiers.
        It can access and update the mechanism's state and the set of credentials public part~$\{c_i^P\}_{i=1}^n$.
    \end{itemize}
\end{definition}

A \emph{player} is defined by its \emph{strategy}, a function that defines its behavior. Formally,
\begin{definition}[Strategy]
    \label{app:strategy}
     A \emph{strategy}~$S^\lambda$ of a player~$p$ is a function, parametrized by the security parameter~$\lambda$,
    that takes a message from the mechanism as input
    (which may be empty, denoted by~$\bot$), and has access to the player's state and credentials~$C_p^\sigma$. 
    It updates the player's state and returns a set of messages to be sent back to the mechanism.
\end{definition}


\subsection{Execution}
\label{app:execution}
We now describe the execution of the system. 
It consists of two parts: \emph{setup} (\setup function, Algorithm~\ref{alg:setup}) and a \emph{main loop} (\mloop function, Algorithm~\ref{alg:main_loop}).
The setup generates a set of~$n$ credentials~$\{c_i=(c_i^P,c_i^S)\}_{i=1}^n$ using the function{~$\gen{M^\lambda}{n}$} (line~\ref{line:generate_creds}).
The credential generation function{~$\gen{M^\lambda}{n}$} draws randomness from the random tape~$v$ and returns a set of~$n$ credentials.
Then the setup function assigns each player credentials according to the scenario~$\sigma$ (lines~\ref{line:assign_user_creds}-\ref{line:assign_attacker_creds}), and the mechanism receives all the credentials' public parts (line~\ref{line:assign_mechanism_creds}).
Then it assigns an identifier to each player:~$\gamma_\textit{ID} \in \{0,1\}$ to the user~(line~\ref{line:player_strategies}) and~$(1 - \gamma_\textit{ID})$ to the attacker~(line~\ref{line:attacker_strategy}).
Whenever a random coin is used by~$M^\lambda$ or the players, the result is the next bit of~$v$. 

Once the setup is complete, the main loop begins~(line~\ref{line:exec:main_loop}).
From this point onwards in the execution, both players~$p\in\{U,A\}$ can send messages to the mechanism, each based on her strategy~$S_p^\lambda$ and the set of credentials available to her~$C_p^\sigma$. 

Time progresses in discrete steps~$t$, accessible to the mechanism and players.
The communication network is reliable but asynchronous.
Messages arrive eventually, but there is no bound on the time after they were sent and no constraints on the arrival order. 
This is implemented as follows.
The scheduler maintains three sets of pending messages, one for each entity identifier~$e \in\{0,1,M^\lambda\}$, denoted by~$\Msgs_e$.
The scheduler is parametrized by an ordering function~$\ord$.

In each step, it uses~$\ord$ to choose a subset of messages (maybe empty) from each of the pending message sets, and returns them as an ordered list (line \ref{line:ord}).
The scheduler removes the chosen messages from the pending messages sets (line~\ref{line:update_msg_sets}).
The function~$\ord$ gets as input the identifier of the entity~$e \in\{0,1,M^\lambda\}$ that the messages are sent to.
It also has access to the state of the execution and a separate scheduler random tape~$v_\gamma$.
However,~$\ord$ does not have access to the messages themselves, the security parameter nor to the entities' random tape~$v$.
We use an ordering function that's oblivious to the messages content and the security parameter to model a system in which, for a sufficiently large security parameter, messages eventually arrive during the execution.
The ordering function's lack of access to the messages and the security parameter ensures that it cannot delay messages until after the execution ends for all~$\lambda$.

Using two separate random tapes, one for the scheduler and another for all other entities in the system, is only for  presentation purpose. 
This is equivalent to using a single global random tape.
Because there exists a mapping between the two, e.g., given a global random tape, we map it into two separate tapes such that the elements in the odd indices are mapped to the scheduler tape and the elements in the even indices are mapped to the other entities' tape.

Each player~$p\in\{U,A\}$ receives the messages chosen by the ordering function~$\ord$ and sends messages to the mechanism according to her strategy~$S_p^\lambda$ (lines~\ref{line:player_msg}-\ref{line:player_msg_end}).
Messages sent by the players are added to the mechanism's pending messages set~$\Msgs_{M}$.
Similarly, the mechanism receives the messages chosen by~$\ord$. 
After every message~$\msg$ it receives, the mechanism updates its state using its function~$\step{M^\lambda}{\msg}$ and returns messages to add to the players' pending messages sets~$\Msgs_0, \Msgs_1$ (line~\ref{line:mechanism_step}).
It then checks if the mechanism has decided which one of the players is recognized as the user by checking its variable~$\textit{decide}_{M^\lambda}$ (lines~\ref{line:decide}-\ref{line:decide_end}). 
Once the mechanism reaches a decision, the execution ends and the player with the matching index (either~$0$ or~$1$) wins.
The tuple~$(\id,\ord, v_\gamma)$ thus defines the \emph{scheduler}'s behavior. 
And an execution~$\exec^\lambda$~(Algorithm~\ref{alg:execution}) is thus defined by its parameter tuple~$(M^\lambda, \sigma, S_U^\lambda, S_A^\lambda, \gamma, v)$; by slight abuse of notation we write~${\exec^\lambda = (M^\lambda, \sigma, S_U^\lambda, S_A^\lambda, \gamma, v)}$.

\subsection{Denial of Service attacks}
\label{app:denial_of_service}
We assume that the attacker cannot affect the delivery time of user messages.
In particular, she cannot perform a denial of service attack by sending messages in a way that causes~$\ord$ to always delay the user's messages.
This is implemented by using two separate ordering functions~$\ord^0, \ord^1$ for the player with identifier~$0$ and the player with identifier~$1$ respectively.

\RestyleAlgo{boxruled}
\begin{algorithm}[t]
        \SetAlgoNoLine 
        \SetAlgoNoEnd 
        \DontPrintSemicolon 
        
    \SetKwInOut{Input}{input}\SetKwInOut{Output}{output}

    

    \nonl \small \emph{$\setup(M^\lambda,\sigma,S_U^\lambda,S_A^\lambda,\gamma,v)$} \scriptsize \;

    $C \gets \gen{M^\lambda}{n}$ \label{line:generate_creds} \tcp*{Generating the set of credentials~$C=\{c_i=(c_i^P,c_i^S)\}_{i=1}^n$}
    $C_U^\sigma \gets \{c_i^S| {\sigma_U}_i =1\}$ \label{line:assign_user_creds} \tcp*{Assigning the credentials' secret parts to the user}
    $C_A^\sigma \gets \{c_i^S| {\sigma_A}_i =1\}$ \label{line:assign_attacker_creds} \tcp*{Assigning the credentials' secret parts to the attacker}
    $C_M^\sigma \gets \{c_i^P| i \in \{1,...,n\}\}$ \label{line:assign_mechanism_creds} \tcp*{Assigning the credentials' public parts to the mechanism} 
    
    $S_{\id}, C_{\id} \gets S_U^\lambda, C_U^\sigma$          \label{line:player_strategies} \tcp*{Assugning~$\id$ to the user}
    $S_{1-\id}, C_{1-\id} \gets S_A^\lambda, C_A^\sigma$          \label{line:attacker_strategy} \tcp*{Assugning~$1-\id$ to the attacker}

    return $M^\lambda, S_0, C_0,  S_1, C_1, v'$

    \caption{Setup of the system, with mechanism~$M^\lambda$ in scenario~$\sigma$}
\label{alg:setup}
\end{algorithm}


\RestyleAlgo{boxruled}
\begin{algorithm}[t]
        \SetAlgoNoLine 
        \SetAlgoNoEnd 
        \DontPrintSemicolon 
        
    \SetKwInOut{Input}{input}\SetKwInOut{Output}{output}

    \nonl \small \emph{$\mloop(M^\lambda,S_0,C_0, S_1, C_1, \ord, v_\gamma, v_\gamma,v,t)$} \scriptsize \;

    $\textit{decide}_{M^\lambda} \gets \bot$\    \label{line:decide_init} \tcp*{Initialize~$M^\lambda$'s decision variable}

    \For{$e \in \{0,1,M^\lambda\}$}{
        $\Msgs_e \gets \emptyset$\        \label{line:init_msg_sets} \tcp*{Initialize the message sets}
    }

    \For{$t=0,1,2,...,\min(t,T(\lambda))$}{
        \For(\tcp*[f]{The scheduler chooses which messages to deliver}){$e \in \{0,1,M^\lambda\}$}{
            $l_e \gets  \ord(e) + [\bot]$\;                    \label{line:ord}         
            $\Msgs_e \gets \Msgs_e \setminus l_e$\;
        }                                                   \label{line:update_msg_sets}

        $\Msgs_{M^\lambda}' \gets \emptyset$\;

        \For{$\msg$ in~$l_0$}{                                 \label{line:player_msg} 
            $\Msgs_{M^\lambda}' \gets \Msgs_{M^\lambda}' \cup S_0(\msg)$                                             \tcp*{Player 0 sends messages to the mechanism}
        }
        \For{$\msg$ in~$l_1$}{                                                         
            $\Msgs_{M^\lambda}' \gets \Msgs_{M^\lambda}' \cup S_1(\msg)$                                                  \tcp*{Player 1 sends messages to the mechanism}
        }                                                   \label{line:player_msg_end}

        \For{$\msg$ in~$l_M$}{                                 \label{line:mechanism_msg}  
            $(\Msgs_0' ,\Msgs_1') \gets \step{M^\lambda}{\msg}$        \label{line:mechanism_step} \tcp*{The mechanism sends messages to the players}
            $\Msgs_0 \gets \Msgs_0 \cup \Msgs_0'$                                                                \tcp*{Player 0's messages are added to the message set}
            $\Msgs_1 \gets \Msgs_1 \cup \Msgs_1'$         \label{line:mechanism_msg_end}                         \tcp*{Player 1's messages are added to the message set}

            \If(\tcp*[f]{Check if the mechanism reached a decision}){$\textit{decide}_{M^\lambda} \neq \bot$}{          \label{line:decide}                          
                $\textit{Return } \textit{decide}_{M^\lambda}$ \label{line:decide_end}                                                     \tcp*{The chosen player wins}
            }
            }                                               
        $\Msgs_{M^\lambda} \gets \Msgs_{M^\lambda} \cup \Msgs_{M^\lambda}'$             \tcp*{The mechanism's messages are added to the message set}

    }
    \caption{Execution main loop of~$M^\lambda$}
\label{alg:main_loop}
\end{algorithm}


\RestyleAlgo{boxruled}
\begin{algorithm}[t]
        \SetAlgoNoLine 
        \SetAlgoNoEnd 
        \DontPrintSemicolon 

    \BlankLine
    $M^\lambda , S_0, C_0, S_1, C_1, v \gets \setup(M^\lambda,\sigma,S_U^\lambda,S_A^\lambda,\gamma, v)$\;
    $\mloop(M^\lambda,S_0, C_0, S_1, C_1,\ord,v_\gamma,v)$\;              \label{line:exec:main_loop}

    \caption{Execution of~$M^\lambda$ in scenario~$\sigma$}
\label{alg:execution}
\end{algorithm}


\subsection{Mechanism Success}
\label{app:model:mechanism_success}

In an asynchronous environment we cannot guarantee that the mechanism correctly identifies the user in any scenario within any bounded time due to the simple fact that the scheduler can delay messages indefinitely. 
However, to maintain cryptographic security, execution time must be bounded---in a longer execution the attacker's probability of guessing credentials is not negligible. 
Previous work~(e.g., \cite{gorilla, safe_permissionless_consensus, Monoxide}) sidesteps this issue with protocols that use an ideal signature scheme that cannot be broken indefinitely. 
We use this abstraction in the rest of the paper.
However, in our case this abstraction is not fully satisfactory,
as we make claims on general mechanisms using standard communication channels, so for any bounded message length we practically cannot rule out forgery. 

Conceptually, we define success as follows.
If the mechanism times out in a given execution, we would like to extend the execution and allow it to succeed in the extended execution. 
To extend the execution without violating cryptographic security, we must increase the security parameter, effectively using a different mechanism. 
Intuitively, we are interested in \emph{mechanism families} that behave similarly for all security parameters, e.g., verify a signature and send a message independent of the signature details. 
The requirement is that given such a scheduler, for all sufficiently large security parameters, mechanisms in the family with sufficiently large parameter succeed. 

Rather than mechanisms, success is therefore defined for mechanism families
that behaves similarly for different security parameters. 
A \emph{mechanism family}~${M}$ is a function that maps a security parameter~$\lambda$ to a mechanism~$M^\lambda$.
Similarly, a \emph{strategy family}~$S$ is a function that maps a security parameter~$\lambda$ to a strategy~$S^\lambda$.
A mechanism family is \emph{successful} in a scenario~$\sigma$ if, for a large enough security parameter, the user wins against all attacker strategies and schedulers. 
Formally, 
\begin{definition}[Mechanism success]
    \label{app:success_family}
    A mechanism family~${M}$ is \emph{successful} in a scenario~$\sigma$, if there exists a user strategy family~$S_U$ 
    such that for all attacker strategy families~$S_A$, schedulers~$\gamma$, and random tapes~$v$, there exists a security parameter~$\lambda_0$ such that:
    \begin{enumerate}
        \item for all security parameters~$\lambda \geq \lambda_0$, the user wins the execution~$\exec^\lambda = (M^\lambda, \sigma, S_U^{\lambda},S_A^{\lambda},\gamma, v)$. And, 
        \item for all security parameters~$\lambda$, either the user wins or the execution times out. 
    \end{enumerate}
    Such a user strategy family~$S_U$ is a \emph{winning user strategy family} in~$\sigma$ with~$M$.
    Otherwise, the mechanism family \emph{fails}. 
\end{definition} 

    

\section{One-Shot Mechanisms Dominance - Full Proof}
\label{app:one_shot_dominance_full_proof}
Given any mechanism family we construct a dominating one-shot mechanism family by simulating the original mechanism's execution. 
\begin{construction} \label{app:con:M_os}
For all~$\lambda \in \mathbb{N}^+$, we define~$M_\os^\lambda$ by specifying the functions~$\gen{M_\os^\lambda}{\cdot}$ and $\step{M_\os^\lambda}{\cdot}$.
The credentials' generation function~$\gen{M_\os^\lambda}{\cdot}$ is the same as~$\gen{M^\lambda}{\cdot}$.
The mechanism's \textit{step} function is described in Algorithm~\ref{alg:M_{OS}}. 

A behavior of~$M_\os^\lambda$ is as follows: 
If it does not receive a message during its execution, it times out.
If it receives multiple messages from a player, it ignores all but the first one~(lines~\ref{line:message_already_recieved}-\ref{line:message_already_recieved_end}).
This is implemented using the variables~$\processing_0$ and~$\processing_1$ that indicate whether the mechanism has received a message from player~$0$ and~$1$, respectively, both initialized to~$0$.
If~$M_\os^\lambda$ receives a message that is not a valid strategy and credentials pair, then it decides the identifier of the other player~(line~\ref{line:invalid_S}).

Consider the first message it receives.
If it is a valid strategy and credentials pair, then~$M_\os^\lambda$ simulates an execution of~$M^\lambda$~(line~\ref{line:simulate_M}) with the given strategy and credentials while setting both the opponent's strategy and credentials to~$\bot$ each~(lines~\ref{line:set_strategy}-\ref{line:set_strategy_end}).
It uses a scheduler random tape~$v_\gamma$ and an execution random tape~$v'$ drawn from~$v$, and an ordering function~$\ord^{v_\gamma}$ that chooses the time and ordering of message delivery randomly based on~$v_\gamma$.
If~$M^\lambda$'s execution decides, then~$M_\os^\lambda$ decides the same value.
Otherwise,~$M_\os^\lambda$ waits for the next message.

If a message arrives from the other player, then similar to the previous case~$M_\os^\lambda$ simulates an execution of~$M^\lambda$ with the given strategy and credentials while setting the opponent's to~$\bot$.
And again, if~$M^\lambda$'s simulated execution decides, then~$M_\os^\lambda$ decides the same value.
Otherwise,~$M_\os^\lambda$ waits until~$T(\lambda)$ execution steps pass and then times out.
So it can decide after receiving the first message from either player or time out.
\end{construction}

The mechanism family~$M_\os(M)$ is one-shot as each of its mechanisms decides based only on the first message it receives from each player.
To prove it dominates~$M$, we~first show that the attacker's message does not lead to her wining.

\begin{lemma}
    \label{app:lem:M_os_attacker_cant_win}
    Let~$\sigma$ be a scenario and let~$M$ be a mechanism family successful in~$\sigma$.
    Then, for all~$\lambda$ and all executions of~$M_\os^\lambda(M^\lambda)$ in scenario~$\sigma$ in which the function~$\step{M_\os^\lambda(M^\lambda)}{\cdot}$ receives a message from the attacker for the first time, either the function sets~$\textit{decide}_{M_\os^\lambda}$ to the user's identifier or the simulated execution of~$M^\lambda$ times out.
\end{lemma}
\begin{proof}
    Let~$\sigma$ be a scenario and let~$M$ be a mechanism family successful in~$\sigma$.
    Let~$\lambda\in \mathbb{N}^+$ and consider an execution of~$M_\os^\lambda$ in scenario~$\sigma$ in which the function~$\step{M_\os^\lambda}{\cdot}$ receives an attacker's message for the first time.
    Denote the identifier of the attacker by~$i\in \{0,1\}$.
    If the attacker's message is not a valid encoding of a strategy and credentials set, then~$M_\os^\lambda$ decides the identifier of the user, and we are done because it sets~$\textit{decide}_{M_\os^\lambda}$ to the user's identifier.
    
    Otherwise, the attacker's message is an encoding of a valid strategy and credentials pair~$(S_A^\lambda,~C_A^\lambda)$.
    In this case,~$M_\os^\lambda$ sets the strategies and credentials to~$S_i=S_A^\lambda$,~$C_i=C_A^\lambda$,~$S_{1-i}= \bot$, and~$C_{1-i}=\bot$,~$\ord^{v_\gamma}$,~$v_\gamma$ and~$v'$ as in the definition of~$M_\os^\lambda$, 
    and simulates~$M^\lambda$'s execution by running~$\mloop(M^\lambda, S_0, C_0, S_1, C_1, \ord^{v_\gamma}, v_\gamma, v')$. 
    If the simulation times-out or returns the identifier of the user we are done.
    The only remaining option is that the simulation returns the identifier of the attacker.
    To show this is impossible, assume by contradiction that the attacker wins in this simulated execution.
    
    First, we show that there exists an execution of~$M^\lambda$ that is identical to the simulated one. 
    Let~${\gamma=(\id^\os,\ord^{v_\gamma}, v_\gamma)}$ be a scheduler such that the user gets the same identifier as in the execution of~$M_\os^\lambda$, with the same ordering function and scheduler random tape that~$M_\os^\lambda$ uses to simulate~$M^\lambda$'s main loop.
    And let~$v_M$ be a random tape such that when the execution~$\exec^\lambda=(M^\lambda,\sigma,\bot,S_A^\lambda,\gamma, v_M)$ reaches the main loop, all the next bits of $v_M$ are equal to~$v'$.
    Note that the main loop of~$\exec^\lambda$ is the same as the one~$M_\os^\lambda$ simulates in~$\exec_\os^\lambda$.
    And the attacker wins in the execution~$\exec^\lambda$ (by the contradiction assumption).

    We thus established a simulated execution in which the attacker wins, and in this simulation there exists a time~$\tau<T(\lambda)$ when the simulated~$M^\lambda$ decides the identifier of the attacker.  
    Since~$M$ is successful in scenario~$\sigma$, there exists a winning user strategy family~$S_U$.
    Let~$\gamma'$ be a scheduler such that in the execution~$\exec'^\lambda = (M^\lambda,\sigma, S_U^\lambda, S_A^\lambda, \gamma',v_M)$ it behaves like~$\gamma$ except~$M^\lambda$ receives the user's messages not before~$\tau$ and before~$T(\lambda$).
    Such a scheduler exists since the communication is asynchronous and message delivery time is unbounded.
    
    At any time step before~$\tau$, the mechanism~$M^\lambda$ sees the same execution prefix whether it is in the execution~$E^\lambda$ with an empty user strategy or in the execution~$E'^\lambda$ with a winning user strategy.
    Thus, it cannot distinguish between the case where it is in~$\exec^\lambda$ or~$\exec'^\lambda$ at~$\tau$. 
    Since in~$\exec^\lambda$ the mechanism~$M^\lambda$ decides at~$\tau$,~$M^\lambda$ must decide the same value at~$\tau$ in~$\exec'^\lambda$.
    That is, the attacker wins also in the execution~$\exec'^\lambda$, contradicting the fact that~$S_U$ is a winning user strategy for~$M$ in~$\sigma$.  
    Thus, the attacker cannot win in the simulated execution of~$M^\lambda$.
\end{proof}

Now we can prove domination.
\begin{proposition}
    \label{app:thm:oneShot}
    For all profiles~$\prof$, an authentication mechanism family that solves the~$\prof$~asynchronous authentication problem is dominated by a one-shot mechanism family. 
\end{proposition}
\begin{proof} 
    Assume that~$M$ is successful in a scenario~$\sigma$.
    Then, there exists a user strategy family~$S_U$ such that for every attacker strategy family~$S_A$, scheduler~$\gamma$, and random tape~$\widetilde{v}$, there exists a security parameter~$\lambda_\suc$, such that for all~$\lambda>\lambda_\suc$, the user wins the corresponding execution~$(M^\lambda, \sigma, S_U^\lambda, S_A^\lambda, \gamma, \widetilde{v})$.
    And for all~$\lambda>0$, either the user wins or the execution times out.

    We now show that~$M_\os$ is successful in scenario~$\sigma$ as well.
    Denote by~$S_U^\os$ the user strategy family such that for all~$\lambda>0$,~$S_U^{\os,\lambda}$ sends an encoding of the strategy~$S_U^\lambda$ from the winning user strategy family~$S_U$ and a set of credentials~$C_U\subseteq C_U^\sigma$ that~$S_U^\lambda$ uses in a single message on the first step.
    And consider any execution of~$M_\os^\lambda$ in scenario~$\sigma$ with the user strategy~$S_U^{\os,\lambda}$.
    Note: In the~$\setup$ function of an execution, the first player in the tuple is always the user and the second is the attacker. 
    However, in the~$\mloop$ function, the first player is the player with identifier~$0$, which can be the user or the attacker, and the second is the player with identifier~$1$.
    
    As for all~$\lambda>0$,~$S_U^{\os,\lambda}$ sends a message in the first step, there exists a security parameter~$\lambda_0^\os$ such that for all~$\lambda^\os>\lambda_0^\os$, the user's message arrives before~$T(\lambda)$. 
    It might be the case that an attacker's message arrives first.
    If the mechanism does not receive any messages, it times out.
    This is possible only if~$\lambda<\lambda_0^\os$, as the user strategy we chose does send a message.
    Otherwise,~$M_\os^\lambda$ eventually receives and processes at least one message during its execution, then we consider two cases separately: a user's message arrives first, or an attacker's message arrives first.

    Denote by~$i\in \{0,1\}$ the identifier of the player~whose message arrives first.
    If the user's message~$\msg$ arrives first with the encoded strategy and credential set~${(S_U^\lambda,~C_U) = \textit{extractStrategy}(\msg)}$,~then~$M_\os^\lambda$~sets the strategies to~$S_i=S_U^\lambda$,~$C_i=C_U$,~${S_{1-i}=\bot}$,~and~$C_{1-i}=\bot$, the scheduler random tape~$\gamma_v$,~the~ordering~function~$\ord^{v_\gamma}$, and the random tape~$v'$ as described in~$M_\os$'s definition,
    and simulates~$M^\lambda$'s execution by~running $\mloop(M^\lambda, S_0, C_0, S_1, C_1, \ord^{v_\gamma},v_\gamma, v')$. 
    Denote by~$\gamma=(\id^\os,\ord^{v_\gamma}, v_\gamma)$ the scheduler such that the user gets the same identifier as in the execution of~$M_\os^\lambda$, with the same ordering function and scheduler random tape that~$M_\os^\lambda$ uses to simulate~$M^\lambda$'s main loop.
    Because~$S_U$ is a winning user strategy family in~$M$, there exists~$\lambda_\suc$ such that for all~$\lambda>\lambda_\suc$, the user wins the execution~$\exec^\lambda = (M^\lambda, \sigma, S_U^\lambda, \bot, \gamma, v')$.
    If~$\lambda>\lambda_\suc$, the simulation terminates and returns the user's identifier~$i$, so~$M_\os^\lambda$ also returns~$i$. 

    Otherwise, as~$\lambda \leq \lambda_\suc$, and because~$S_U$ is a winning user strategy family for~$M^\lambda$, if the simulation terminates then it must return the identifier of the user.
    And so does the mechanism~$M_\os^\lambda$.
    If the simulation does not terminate, then~$M_\os^\lambda$ waits for the next message from the other player (the attacker).
    If no message arrives before~$T(\lambda)$, then~$M_\os^\lambda$ times out as well (note this is possible only if~$\lambda\leq \lambda_\suc$ because otherwise~$M_\os^\lambda$ would have already returned the user's identifier).
    Otherwise,~$M_\os^\lambda$'s step function receives the attacker's message, and by Lemma~\ref{lem:M_os_attacker_cant_win}, either~$M_\os^\lambda$ decides the user or the simulated execution of~$M^\lambda$ times out.
    If the simulation times out, then~$M_\os^\lambda$ times out as well.
    Overall we showed that if~$\lambda>\lambda_0^\os$, a user's message is received by~$M_\os^\lambda$.
    And if the user's message arrives first, then if~$\lambda>\lambda_\suc$, the user wins, and for all~$\lambda$, either the user wins or the execution times out.

    Now assume the attacker's message arrives first to~$M_\os^\lambda$.
    Again by Lemma~\ref{lem:M_os_attacker_cant_win}, either~$M_\os^\lambda$ decides the user or the simulated execution of~$M^\lambda$ in~$M_\os^\lambda$ times out.
    If~$M_\os^\lambda$ decides the user, then we are done.
    Otherwise, the simulation times out, and~$M_\os^\lambda$ waits for the next message from the other player (the user).
    If no other message arrives before~$T(\lambda)$, then~$M_\os^\lambda$ times out. 
    This is possible only if~$\lambda<\lambda_0^\os$ by definition of~$\lambda_0^\os$.

    Otherwise,~$M_\os^\lambda$ receives the user's message with her encoded strategy and credentials set. 
    Then similar to the previous case,~$M_\os^\lambda$ sets the strategies, credential sets, scheduler random tape, ordering function, and random tape as described in Construction~\ref{con:M_os},
    and simulates~$M^\lambda$'s execution by running~$\mloop(M^\lambda, S_0, C_0, S_1, C_1, \ord^{v_\gamma},v_\gamma, v')$. 
    By the same argument as before, we get that if~$\lambda>\lambda_\suc$, the simulation terminates and returns the user's identifier, so~$M_\os^\lambda$ also returns it. 
    Otherwise,~$\lambda\leq \lambda_\suc$, if the simulation terminates, it must return the identifier of the user.
    And so does the mechanism~$M_\os^\lambda$.
    If the simulation does not terminate, then~$M_\os^\lambda$ times out.

    Overall we showed that if~$\lambda>max(\lambda_\suc, \lambda_0^\os)$, the user wins. 
    And for~$\lambda \leq max(\lambda_\suc, \lambda_0^\os)$, either the user wins or the execution times out.
    Therefore,~$M_\os$ is successful in scenario~$\sigma$.
    We thus conclude that the one shot mechanism family~$M_\os$ dominates $M$.
\end{proof}

\section{Deterministic Mechanisms are Dominated by Boolean Mechanisms}
\label{app:deterToBoolCreds}
\begin{restate}{Lemma~\ref{lem:deterToBoolCreds}}
    \lemmaDeterToBool
\end{restate}

\begin{proof}
    \sloppy
    Let~$M_\deter$ be a deterministic mechanism and let~$f$ be as defined in Construction~\ref{con:M_f}.
    We prove that if~$M_{\deter}$ is successful in a scenario, then~$M_{f}$ is successful as well.
    Assume~$M_\deter$ is successful in scenario~$\sigma=(\sigma_U,\sigma_A)$, by definition of~$f$, we get that~$f(\sigma_U)=1$.

    As~$M_\deter^\lambda$ is successful in~$\sigma$, there exists a winning user strategy family~$S_U^\deter$ for~$M_\deter$.
    And because~$M_\deter$ is a credential-based mechanism, for each~$\lambda$,~$S_U^{\deter,\lambda}$ sends a set of credentials~$c_U^\lambda$ in a single message. 
    Let~$\lambda \in \mathbb{N}^+$ be any security parameter, let~$S_U^{f, \lambda}$ be the user strategy for~$M_f^\lambda$ that sends the same subset of the user's credentials as~$S_U^{\deter,\lambda}$ in a single message in the first step.
    Let~$S_A^f$ be an attacker strategy family for~$M_f$,~$\gamma_f$ a scheduler, and~$v$ a random tape.
    
    Consider the execution~$\exec_f^\lambda = (M_f^\lambda, \sigma, S_U^{f,\lambda}, S_A^{f,\lambda}, \gamma_f, v)$.
    Because~$S_U^{f,\lambda}$ sends a message in the first step, there exists a security parameter~$\lambda_0^f$ such that for all~$\lambda>\lambda_0^f$, the user's message arrives before~$T(\lambda)$. 
    It might be the case that an attacker's message arrives first.
    If the mechanism does not receive any messages, it times out.
    This is possible only if~$\lambda\leq \lambda_0^f$, by definition of~$\lambda_0^f$.

    Otherwise,~$M_f^\lambda$ eventually receives and processes at least one message during its execution, then we consider two cases separately: a user's message arrives first, or an attacker's message arrives first.
    If the user's message arrives first to~$M_f^\lambda$, with her set of credentials, then as~$f(\sigma_U)=1$ (because~$\sigma \in \Prof{M_\deter}$),~$M_f^\lambda$ decides the user.
    We get that if~$\lambda>\lambda_0^f$, then~$M_f^\lambda$ decides the user.
    And if~$\lambda\leq\lambda_0^f$, then~$M_f^\lambda$ either decides the user or times out.
    Overall we showed that if the user's message arrives first, then if~$\lambda>\lambda_0^f$ the user wins, otherwise, either the user wins or the execution times out.
    So~$\lambda_\suc^f=\lambda_0^f$.

    Otherwise, the attacker's message arrives first to~$M_f^\lambda$.
    If~$f(\sigma_A)=0$, then~$M_f^\lambda$ decides the user, and we are done, because the user wins.
    The only other option is if~${f(\sigma_A)=1}$.
    We show that this is not possible.

    By definition of~$f$,~$f(q)=1$ if and only if there exists a scenario~$\sigma_{U:q}$ in which the user's availability vector is~$q$ and~$M_\deter$ succeeds in~$\sigma_{U:q}$.
    Let~$S_{U:q}$ be the winning user strategy for~$M_\deter$ in~$\sigma_{U:q}$.
    Then for all attacker strategy families~$S_A$, schedulers~$\gamma = (\id,\ord,v_\gamma)$ and random tapes~$v$, there exists a security parameter~$\lambda_\suc$ such that for all~$\lambda>\lambda_\suc$, the user wins the execution~$\exec_{U:q}^\lambda = (M_\deter^\lambda, \sigma_{U:q}, S_{U:q}^\lambda, S_A^\lambda, \gamma, v)$.
    This includes the attacker strategy that does not send any messages~$S_A^\lambda = \bot$.
    Then during the execution~$\exec_{U:q}^\lambda$'s main loop there exists a time~$\tau<T(\lambda)$ when~$M_\deter^\lambda$ decides the identifier of the user.  

    Now consider the execution of~$M_\deter^\lambda$ in~$\sigma$, in which the attacker uses the same strategy as the users'~$S_{U:q}^\lambda$ from~$\sigma_{U:q}$.
    This is possible as the attacker has access to all credentials the user had in~$\sigma_{U:q}$.
    Let~$\gamma'$ be a scheduler in which the user and the attacker's identifiers are opposite to the ones in~$\gamma$, the attacker's message arrives first and the user's message is delayed beyond~$\tau$.
    And let~$v'$ be any random tape.
    Denote this execution by~$\exec^\lambda$. 

    At any time step before~$\tau$, the mechanism~$M_\deter^\lambda$ sees the same execution prefix whether it is in~$\exec^\lambda$ or in~$\exec_{U:q}^\lambda$, 
    and thus, it cannot distinguish between the case where it is in~$\exec^\lambda$ or~$\exec_{U:q}^\lambda$. 
    As~$M_\deter^\lambda$ decides the user in~$\exec_{U:q}^\lambda$'s main loop at time~$\tau$, then it must decide the attacker at the same time~$\tau$ in~$\exec^\lambda$'s main loop (as the identifiers are reversed, and the mechanism sees the exact same execution prefix). 
    Therefore,~$M_\deter^\lambda$ decides the attacker in~$\exec^\lambda$ contradicting the assumption that~$M_\deter^\lambda$ succeeds in~$\sigma$.
    We conclude that~$f(\sigma_A)=0$ and~$M_f^\lambda$ decides the user also in this case. 
    Therefore,~$M_{\textit{f}}$ is successful in~$\sigma$.
    Overall we showed that the mechanism family~$M_{\textit{f}}$ a monotonic Boolean mechanism family and dominates~$M_{\textit{det}}$.
\end{proof}

\section{Success Equivalence}
\label{app:proof_success_equivalence}

\begin{lemma}
    Let~$f$ be a monotonic Boolean function, and let~$M_f$ be the mechanism family of~$f$.
    $M_f$ is successful in a scenario~$\sigma$ if and only if~$f(\sigma_U)=1$ and~$f(\sigma_A)=0$.
\end{lemma}

\begin{proof}
    For the first direction, assume that~$M_f$ is successful in a scenario~$\sigma$.
    Then, there exists a winning user strategy family~$S_U$ for~$M_f$ in~$\sigma$ that sends her set of credentials in a single message on the first step.
    For every attacker strategy~$S_A$, scheduler~$\gamma$ and random tape~$v$, there exists a security parameter~$\lambda_\suc$ such that for all~$\lambda>\lambda_\suc$ the user wins the execution~$(M_f^\lambda, \sigma, S_U, S_A ,\gamma, v)$.
    
    Because~$S_U^\lambda$ sends a message in the first step, there exists a security parameter~$\lambda_0^f$ such that for all~$\lambda>\lambda_0^f$, the user's message eventually arrive before~$T(\lambda)$. 
    It might be the case that an attacker's message arrives first.
    If the mechanism does not receive any messages, it times out.
    This is possible only if~$\lambda<\lambda_0^f$, as the user strategy we chose always send a message.

    Let~$\lambda> \lambda_0^f$ and consider a scheduler~$\gamma$ for which the user's message arrives first.
    Consider the execution~$\exec = (M_f^\lambda, \sigma, S_U, S_A ,\gamma, v)$.
    As the user's message arrives first, with her set of credentials,~$M_f$ will extract the availability vector~$\sigma_U$ of the credentials. 
    Then, as~$M_f$ is successful in~$\sigma$, it will return the identifier of the user.
    This can happen only if~$f(\sigma_U)=1$.
    
    Now let~$S_U$ be a winning user strategy and~$S_A$ an attacker strategy that sends any subset of the attacker's credentials.
    Let~$\gamma'$ and~$v'$ be a scheduler and random tape such that the attacker's message arrives first, and consider the execution~${\exec' = (M_f^\lambda, \sigma, S_U, S_A ,\gamma', v')}$.
    In that case, as the attacker sends a subset of her credentials,~$M_f^\lambda$ will extract the availability vector~$\sigma_A'$ of the credentials. 
    Then, as~$M_f$ is successful in~$\sigma$, it will return the identifier of the user.
    This is possible only if~$f(\sigma_A')=0$.
    This is true for any subset of credentials the attacker chooses, including the set of all her credentials. 
    Thus,~$f(\sigma_A)=0$.
    Therefore, we showed that if~$M_f$ is successful in a scenario~$\sigma$, then~$f(\sigma_U)=1$ and~$f(\sigma_A)=0$.
    
    For the second direction, assume that~$f(\sigma_U)=1$ and~$f(\sigma_A)=0$.
    Let~$S_U$ be a user strategy family that sends the set of credentials corresponding to~$\sigma_U$ in the first step.
    And let~$S_A$,$\gamma$ and~$v$ be any attacker strategy, scheduler and random tape.
    We show that~$S_U$ is a winning user strategy for~$M_f$ in~$\sigma$.

    Consider the execution~${(M_f^\lambda, \sigma, S_U, S_A ,\gamma, v)}$.
    If no message was received, then~$M_f^\lambda$ times out.
    This is possible only if~$\lambda<\lambda_0^f$, as the user strategy we chose sends a message on the first step.
    Otherwise,~$M_f^\lambda$ receives a message.

    If the user's message arrives first to~$M_f^\lambda$ with the credentials corresponding to~$\sigma_U$, then~$M_f^\lambda$ extracts the availability vector~$\sigma_U$ of the credentials. 
    Then, as~$f(\sigma_U)=1$, the mechanism~$M_f^\lambda$ returns the identifier of the user, thus the user wins the execution.
    Otherwise, the attacker's message~$\msg_A$ arrives first to~$M_f^\lambda$.
    If~$\msg_A$ does not contain a valid set of credentials, then~$M_f^\lambda$ returns the identifier of the user, and she wins the execution.
    Otherwise,~$M_f^\lambda$ extracts the availability vector~$q \leq \sigma_A$ of the credentials in~$\msg_A$ (the attacker might send any subset of credentials she knows). 
    Then, because~$f$ is monotonic and~$f(\sigma_A)=0$, also~$f(q)=0$ and $M_f^\lambda$ returns the identifier of the user, thus the user wins the execution.
    In both cases, if~$\lambda>\lambda_0^f$, the execution terminates and the user wins.
    And if~$\lambda < \lambda_0^f$, either the user wins or the execution times out.
    Thus,~$M_f$ is successful in~$\sigma$, concluding the proof that~$M_f$ is successful in a scenario~$\sigma$ if and only if~$f(\sigma_U)=1$ and~${f(\sigma_A)=0}$.
\end{proof}

\section{Partial Boolean Mechanisms}
\label{app:proof_partial_Boolean_function}
\begin{restate}{Observation~\ref{clm:partial_Boolean_function}}
    \partialBooleanfunction
\end{restate}

\begin{proof}
    Let~$M_1$,~$M_2$,~$T_1$,~$T_2$,~$F_1$ and~$F_2$ be as in the claim.
    We show that,~$T_1 \subseteq T_2$ and~$F_1 \subseteq F_2$.
    First, note that as~$\Prof{M_1} \neq \emptyset$, then~$T_1$ and~$F_1$ are both not empty.
    Because the profile of a mechanism is the Cartesian product of the set of user availability vectors~($T_1$) with the set of attacker availability vectors~($F_1$)~(Observation~\ref{cor:profile_of_M_f}).
    Similarly,~$T_2$ and~$F_2$ are both not empty.
    Assume for contradiction this  at least one of the following holds:~$T_1 \not\subseteq T_2$ or~$F_1 \not\subseteq F_2$.
    
    Without loss of generality, assume that~$T_1 \not\subseteq T_2$.
    As~$\Prof{M_1}\neq \emptyset$, there exists~${q\in T_1}$ such that~$q\notin T_2$.
    Consider a scenario~$\sigma$ such that~$\sigma_U=q$ and~$\sigma_A\in F_1$.
    Then,~$\sigma\in \Prof{M_1}$ and~$\sigma\notin \Prof{M_2}$, contradicting the fact $\Prof{M_1} \subseteq \Prof{M_2}$.
    Thus,~$T_1 \subseteq T_2$.
    Similarly,~$F_1 \subseteq F_2$.
\end{proof}

\begin{restate}{Lemma~\ref{lem:profile_partial_func}}
    \profilePartialFnc
\end{restate}

\begin{proof}
    Let~$M$ be a mechanism with~$n$ credentials.
    If~$M$'s profile is empty, then~$M$ is equivalent to the Boolean mechanism of constant Boolean function~$f(x)=0$.
    A constant function is monotonic, and thus we are done.
    
    Otherwise, by Theorem~\ref{thm:M_to_inc}, there exists a monotonic Boolean function~$f$ such that the Boolean mechanism~$M_f$ dominates~$M$.
    That is,~$\Prof{M} \subseteq \Prof{M_f}$.
    Let~$T$ be the set of all user availability vectors in~$\Prof{M}$ and~$F$ be the set of all attacker availability vectors in~$\Prof{M}$.

    By definition, To prove that~$(T,F)$ is a partial Boolean function, we must show that~$T\cap F = \emptyset$.
    Let~$T_f$ and~$F_f$ be the sets of all user and attacker availability vectors in~$\Prof{M_f}$ respectively.
    As~$M$ and~$M_f$'s profiles are not empty, by Observation~\ref{clm:partial_Boolean_function}, we get that~$T\subseteq T_f$ and~$F\subseteq F_f$.

    We show that~$T\cap F = \emptyset$.
    If there exists a vector ${q\in T\cap F\subseteq T_f \cap F_f}$, then for~$q$, it holds that~$f(q)=1$ and~$f(q)=0$, contradicting the fact that~$f$ is a well-defined function.
    Thus,~$T \cap F=\emptyset$ and~$(T,F)$ is a partial Boolean function.
    It is monotonic because~$f$ is an extension of~$(T,F)$ and~$f$ is monotonic.
    And the mechanism~$M$ is equivalent to the monotonic partial Boolean mechanism of~$(T,F)$.
\end{proof}

\section{Profile Size Bounds}
\label{app:profile_size_bound}

We calculate the number of viable scenarios.
\begin{restate}{Observation~\ref{obs:non_viable_scenarios_num}}
    \label{app:obs:non_viable_scenarios_num}
    \numViableScenarios
\end{restate}

\begin{proof}
    The number of scenarios for~$n$ credentials is~$4^n$.
    The number of viable scenarios is the number of scenarios in which there exists at least one safe credential.
    The number of scenarios in which there are no safe credentials is the number of scenarios in which all credentials are either lost, leaked, or stolen.
    That is,~$3^n$ scenarios.
    Therefore, the number of viable scenarios is~$4^n-3^n$.
\end{proof}

We prove an upper bound on the numbe of scenarios that can be added to a profile of a partial Boolean function.
\begin{restate}{Observation~\ref{obs:scenarios_addition}}
    \profileSizeBound
\end{restate}

\begin{proof}
    Let~$(T,F)$ be a partial Boolean function of~$n$ credentials and let~$s = max(|T|,|F|)$.
    For~$n$ credentials, the number of Boolean vectors is~$2^n$.
    By Observation~\ref{cor:profile_of_M_f}, we have~$|\Prof{M_{(T,F)}}| = |T| \cdot |F|$.
    There are three separate cases:
    
    (1) If~$s \geq 2^{n-1}$ and~$s=|T|$ then the maximal size of a profile is reachable if all vectors~$v\notin T\cup F$ are added to~$F$.
    That is, the number of scenarios that can be added to the profile is at most~${|T| \cdot (2^n-|T|)- |\Prof{M_{(T,F)}}|}$
    
    (2) If~$s \geq 2^{n-1}$ and~$s=|F|$ the case is symmetric to the previous one and the number of scenarios that can be added to the profile is at most~${|F| \cdot (2^n-|F|)- \Prof{M_{(T,F)}}}$
    
    (3) Otherwise, we get that $s < 2^{n-1}$, then the maximal size of a profile is~$4^n-3^n$, the same as the number of viable scenarios (Observation~\ref{obs:non_viable_scenarios_num}) as by \cite{cryptoeprint:2022/1682}, non-viable scenarios are not in the profile.
    In this case, the number of scenarios that can be added to the profile is at most $4^n-3^n - |\Prof{M_{(T,F)}}|$.
\end{proof}

\section{Algorithm Analysis}
\label{app:algorithm_analysis}

\subsection{Algorithm Correctness}
\label{app:algorithm_correctness}
We first prove that our algorithm keeps the (possibly partial) truth table monotonic after each update.
\begin{claim}
    \label{app:claim:monotonicity_update}
    After each update, the partial monotonic truth table remains a (possibly partial) monotonic truth table, and the above update does not contradict previous ones.
\end{claim}

\begin{proof}
    Given a monotonic partial truth table, it holds that for all~$x,y\in\{0,1\}^n$ such that~$x>y$ we have that if~$f(y)=1$ then~$f(x)=1$, and if~$f(x) = 0$ then~$f(y)=0$.
    According to the algorithm, in every step where we add a viable scenario~$\sigma$, we check that~$\sigma_U$ is not set to~$0$ and~$\sigma_A$ is not set to~$1$.
    If so, we set~$\sigma_U$ to~$1$ and~$\sigma_A$ to~$0$.
    As~$\sigma$ is viable, we have that~$\sigma_U \nleq \sigma_A$, so this update preserves~monotonicity.
    
    Because~$\sigma_U$ was not set to~$0$ before, for all~$x>\sigma_U$ we have that~$f(x) \neq 0$, and we can set them to~$1$.
    Similarly, because~$\sigma_A$ was not set to~$1$ before, for all~${y<\sigma_A}$,~$f(y) \neq 1$, and we can set them to~$0$.
    Therefore, the truth table can still be completed into a monotonic Boolean function.
\end{proof}

We now prove our algorithm correctness, including termination, the truth table validity, and bound the distance of the returned mechanism from the optimal one.

\begin{lemma}
    \label{app:lemma:algorithm_correctness}
    Let~$n$ be the number of credentials and let~$0 \leq \delta < 1$.
    Then when Algorithm~\ref{alg:scenario_based_search} returns, the variable \maxTable contains a truth table of a mechanism that is at least~$\delta$ close to the optimal mechanism.
\end{lemma}

\begin{proof}
    We divide the proof into three parts.
    First, we show that the algorithm always stops.
    Second, we show that when the algorithm returns, \maxTable contains a valid complete monotonic truth table.
    Finally, we show that the mechanism defined by the truth table is at least~$\delta$ close to the optimal mechanism.
    
    To show that the algorithm always stops, we note that the number of scenarios that can be added to a profile of a partial Boolean mechanism is bounded.
    Thus, in the worst case, the algorithm explores all possible scenarios, which is a finite number. 
    Once no scenarios are left, the condition in line~\ref{alg:line:zero_additions} is met, and the algorithm stops.
    
    To show that the algorithm always results in a valid complete monotonic truth table, 
    we note that by Claim~\ref{app:claim:monotonicity_update}, the truth table remains a (possibly partial) monotonic truth table after each update.
    Thus, it is sufficient to show that the algorithm updates \maxTable to a valid complete monotonic truth table at least once, 
    and never updates it to a non-complete truth table.

    The algorithm stops the recursive exploration only when one of the three stopping conditions is met.
    In the first case (line~\ref{alg:line:complete_table}), the truth table is complete, and the algorithm updates \maxTable to this complete table.
    In the second case (line~\ref{alg:line:zero_additions}), the algorithm completes the truth table arbitrarily, and \maxTable gets the completed table.
    Finally, in the third case (line~\ref{alg:line:delta_condition}), the algorithm prunes the branch, and returns, without updating \maxTable. 
    Thus, if the algorithm updates \maxTable, it always updates it to a complete valid truth table.

    To show that the algorithm updates \maxTable at least once, 
    we note that the third stopping condition (line~\ref{alg:line:delta_condition}) is the only one that does not update \maxTable.
    However, this condition is met only if \maxSuccessProb is greater than 0.
    And as \maxSuccessProb is initialized to 0, and \maxTable and \maxSuccessProb are updated only together (lines~\ref{alg:line:max_update_1}-\ref{alg:line:max_update_1_2} and~\ref{alg:line:max_update_2}-\ref{alg:line:max_update_2_2}),
    the algorithm must return at least once because one of the first two stopping conditions is met.
    Thus, \maxTable is updated to a valid complete monotonic truth table.

    To complete the proof, we show that the mechanism defined by the truth table's success probability is at most~$\delta$ lower than the optimal mechanism.
    Denote by~$M$ the mechanism defined by the truth table \maxTable after the algorithm returns, with success probability \maxSuccessProb, and by~$M^*$ the optimal mechanism.
    Assume for contradiction that the probability $\maxSuccessProb +\delta$ is smaller than the success probability of~$M^*$.
    By Theorem~\ref{thm:M_to_inc}, the optimal mechanism~$M^*$ can be represented by a (complete) monotonic Boolean function.
    Therefore, there exists a path along the recursive exploration tree that leads to the optimal mechanism, which the algorithm pruned (in line~\ref{alg:line:delta_pruning}).
    
    Consider the function call in which the algorithm pruned the path to the optimal mechanism.
    From Observation~\ref{obs:scenarios_addition}, 
    we know that~$\textit{\numPossibleAdditions}$ (line~\ref{alg:line:num_possible_additions}) is an upper bound on the number of scenarios that can be added to the current truth table without contradicting it.
    As \potentialSuccessProb is calculated by summing the probabilities of the current partial truth table and the next~$\textit{\numPossibleAdditions}$ scenarios with the highest probabilities,
    then in this branch, $\potentialSuccessProb \geq M^*$'s success probability.
    However, as the algorithm pruned this branch, we know that $\maxSuccessProb > \potentialSuccessProb - \delta$ (line~\ref{alg:line:delta_condition}).
    Therefore, $\maxSuccessProb + \delta > M^*$'s success probability.
    Contradicting the assumption that $\maxSuccessProb + \delta$ is smaller than the success probability of~$M^*$.
    Therefore, the algorithm returns a mechanism that is at most~$\delta$ away from the optimal mechanism.
   
    Overall, we showed that the algorithm returns a valid complete monotonic truth table representing a mechanism whose success probability is at most~$\delta$ lower than the optimal mechanism.
\end{proof}

\subsection{Algorithm Complexity}
\label{app:sec:algo_complexity}
We first discuss the difficulty of calculating the exact complexity of the algorithm in a general case~(\S\ref{sec:disscusion_complexity}).
Then we evaluate the complexity of the algorithm empirically~(\S\ref{sec:empirical_complexity}), by measuring the runtime of our algorithm for different numbers of credentials and fault probabilities.

\subsubsection{Bound}
\label{app:sec:disscusion_complexity}
To give a loose upper bound on a single function call of the algorithm's complexity, we analyze the complexity of each step separately.
First, to calculate the current profile (line~\ref{alg:line:get_profile}), we use a Cartesian product of the availability vectors of the user and the attacker as described in Observation~\ref{cor:profile_of_M_f}.
By~\cite{cryptoeprint:2022/1682} and Observation~\ref{obs:non_viable_scenarios_num}, this is bounded by~$O(4^n-3^n)$. 
Calculating the success probability of the current truth table (line~\ref{alg:line:curr_success}) requires summing the probabilities of the scenarios in the profile ($O(4^n-3^n)$).
Finding and sorting the possible additional scenarios probabilities (line~\ref{alg:line:sorted_additional_probabilities}) is~$O((4^n-3^n)\cdot n)$.
Finding the number of possible additions (line~\ref{alg:line:num_possible_additions}) is calculated based on Observation~\ref{obs:scenarios_addition} and requires~$O(1)$.
The bound itself is~$O(4^n-3^n)$ in the worst case, if we could add all of the viable scenarios.
Therefore, calculating the sum of the probabilities of the next~$\textit{\numPossibleAdditions}$ scenarios (line~\ref{alg:line:max_possible_additions}) is~$O(4^n-3^n)$.
Excluding the function calls, all other steps are~$O(1)$.
Thus, the complexity of a single recursive function call is~$O((4^n-3^n)\cdot n)$.
And in the worst case, the recursion depth may reach~$O(2^{4^n-3^n})$.
Resulting in a total complexity of~$O(2^{4^n-3^n}\cdot (4^n-3^n)\cdot n)$.

But this analysis assumes the worst case for every step of the algorithm.
However, in practice, many steps' worst case cannot happen simultaneously.
Thus, the actual complexity of the algorithm is much lower than the calculated upper bound.

Calculating the exact complexity of the algorithm depends on the number of credentials, their specific probabilities, and the parameter~$\delta$.
While an upper bound on the recursion depth is given by 2 to the power of the number of viable scenarios, the actual depth explored is much smaller.
This is due to multiple factors, for example (1) the exponential drop in the probability of different scenarios, combined with the fact that the algorithm prunes branches with negligible advantage;
(2) the fact that the number of scenarios that can be added to a profile of a partial Boolean mechanism is limited~(Observation~\ref{obs:scenarios_addition})
in a non-trivial way depending on which do result in a monotonic mechanism;
and (3)~each scenario adds not a single row to the truth table, but possibly many rows for keeping monotonicity 
(depending on the specific scenario and the current truth table).
While all those factors contribute to the algorithm's efficiency, 
it remains an open question whether the algorithm's exact complexity can be calculated theoretically (cf. the lack of a closed form expression for the number of monotinic Boolean functions~\cite{Dedekind1897}).

\subsubsection{Empirical Complexity}
\label{app:sec:empirical_complexity}

\begin{figure}[t]
    \centering
    \includegraphics[width=0.47\textwidth]{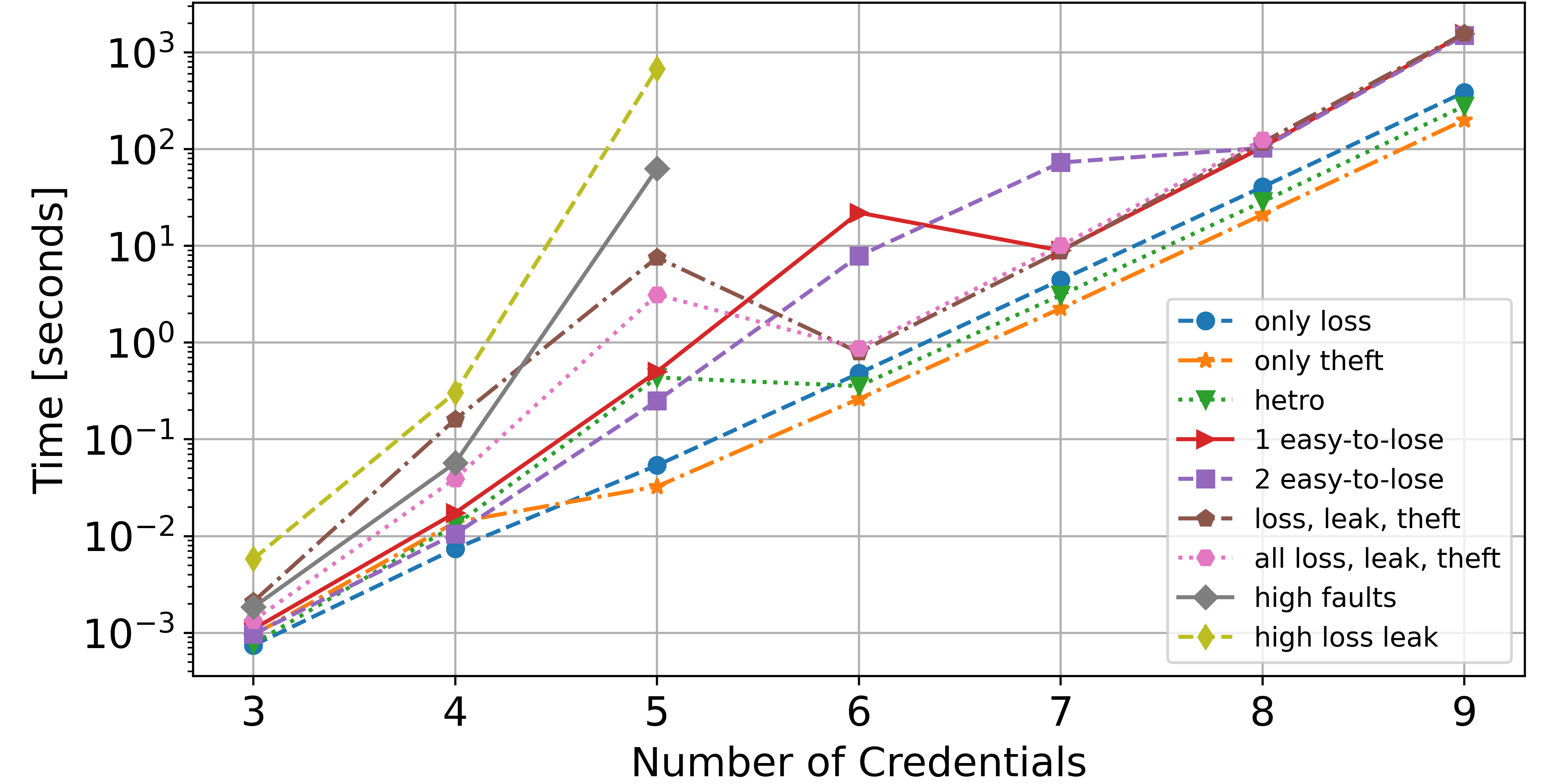}
    \caption{Runtime of the algorithm as a function of the number of credentials for different fault probabilities with $\delta=10^{-5}$.}
    \label{fig:runtime_full}
\end{figure}



We provide additional examples of our algorithm runtime behavior in Table~\ref{app:tab:time_functions} and Figure~\ref{fig:runtime_full}. 

Now we describe each plot in Figure~\ref{fig:runtime_full} and the corresponding fitted functions in Table~\ref{app:tab:time_functions}, categorized by the runtime function behavior.

\textbf{Exponential Growth:}
When one credential can suffer from up to a single type of fault and the rest can have up to two types of faults,
or when all credentials can suffer from all types of faults with low probability,
the algorithm's runtime grows exponentially ($O(4^n)$) with the number of credentials (all fits with~${R^2 > 0.97}$, exact values are in Table~\ref{app:tab:time_functions}).

For example, when all credentials can suffer from loss with~$P^\loss = 0.01, P^\leak=P^\theft=0$ 
(\emph{only loss} in Figure~\ref{fig:runtime_full}), 
or all credentials can suffer from theft with~$P^\theft = 0.01, P^\loss=P^\leak=0$ 
(\emph{only theft} in Figure~\ref{fig:runtime_full}).

When one key is prune to loss and leak with~$P^\loss=P^\leak=0.01, P^\theft=0$ and the rest can only be stolen with~$P^\theft=0.01$ 
(\emph{hetro} in Figure~\ref{fig:runtime_full}).
Or when one credential can be easily lost, but not leaked or stolen, with~$P_1^\loss=0.3, P_1^\leak=P_1^\theft=0$ and the rest can be either lost or leaked with~$j>1$,~$P_j^\loss=P_j^\leak=0.01,P_j^\theft=0$. 
(\emph{1~easily-to-lose} in Figure~\ref{fig:runtime_full}).
Similarly, when 2 credentials can be easily lost, but not leaked or stolen 
(\emph{2~easily-to-lose} in Figure~\ref{fig:runtime_full}).

We observe a similar trend when each credential can have all three types of faults, where at least two with low fault probabilities, $P^\loss=0.01, P^\leak=P^\theft=0.001$
(\emph{loss, leak, theft} in Figure~\ref{fig:runtime_full}),
or when one credential can be only lost with~$P_1^\loss=0.01, P_1^\leak=P_1^\theft=0$ and the rest can have all three types of faults with~$i>1$,~$P_i^\loss=0.01, P_i^\leak=P_i^\theft=0.001$ 
(\emph{all loss, leak, theft} in Figure~\ref{fig:runtime_full}).

\textbf{Super-Exponential Growth:}
When all credentials can suffer from two or more types of faults with a high probability, the algorithm's runtime grows super-exponentially with the number of credentials.
E.g., if all three fault types are possible with high probabilities, say $P^\loss=0.1, P^\leak=0.3, P^\theft=0.4$ (\emph{high faults} in Figure~\ref{fig:runtime_full}). 
Or when all credentials are pruned to two or more types of faults with high probabilities (e.g., $P^\loss=P^\leak=0.01, P^\theft=0$) (\emph{high loss leak} in Figure~\ref{fig:runtime_full}),
the algorithm's runtime grows too rapidly to obtain sufficient data points for regression analysis.

\begin{table*}[t]
    \centering
    \caption{Runtime Analysis}
    \begin{tabular}{|c|c|c|c|c|c|}
    \hline
    \textbf{Name} & \textbf{Loss} & \textbf{Leak} & \textbf{Theft} & \textbf{Formula} & \textbf{$R^2$} \\ \hline
    only loss & {$10^{-2}$} & {$0$} & {$0$} & $0.0014 \cdot 4^n  -12.01$ & 0.979 \\ \hline
    only theft & {$0$} & {$0$} & {$10^{-2}$} & $0.0007 \cdot 4 ^ n -6.20$ & 0.979 \\ \hline
    hetro & {$P_1=10^{-2}$, $i >1: P_i =0$} & {$P_1= 10^{-2}$, $i>1, P_i =0$} & {$P_1 = 0$, $i>1, P_i =10^{-2}$} & $0.001 \cdot 4 ^ n-8.68$ & 0.979 \\ \hline
    1easy-to-lose & {$P_1$, $i >1: P_i =10^{-2}$} & {$P_1 = 0$, $i>1, P_i =10^{-2}$} & {$0$} & $0.0059 \cdot 4 ^ n + -53.54$ & 0.967 \\ \hline
    2easy-to-lose & {$P_1=P_2=0.3$, $i >2: P_i =10^{-2}$} & {$P_1 = P_2 = 0$, $i>2, P_i =10^{-2}$} & {$0$} & $0.0056 \cdot 4^ n -42.99$ & 0.968 \\ \hline
    loss, all-leak & {$P_1=0$, $i >1: P_i =10^{-2}$} & {$10^{-2}$} & {$0$} & $0.0004 \cdot 4 ^ n -0.41$ & 0.977\\ \hline
    loss, leak, theft & {$10^{-2}$} & {$10^{-3}$} & {$10^{-3}$} & $0.006 \cdot 4^n -55.44$ & 0.97 \\ \hline
    all loss, leak, theft & {$10^{-2}$} & {$P_1= 0$, $i>1, P_i =10^{-3}$} & {$P_1= 0$, $i>1, P_i =10^{-3}$} & $0.0019 \cdot 4^n -4.69$ & 0.97 \\ \hline
\end{tabular}
\label{app:tab:time_functions}
\end{table*}

\end{document}